\documentclass{llncs}

\usepackage[table]{xcolor} %For gray color in Table: https://tex.stackexchange.com/questions/8891/color-merged-and-regular-cells-in-a-table-individually
\usepackage[utf8]{inputenc}
\usepackage{amsmath}
\usepackage{amssymb}
\usepackage{adjustbox}
\usepackage{pdfpages}
\usepackage{caption}
\usepackage[ruled,vlined]{algorithm2e}
\usepackage{upgreek}
\usepackage{latexsym }
\usepackage{ stmaryrd }
\newcommand{\comment}[1]{}
\usepackage{centernot,cancel}
\usepackage[toc,page]{appendix}
\usepackage{tabto}
\usepackage{comment}
\usepackage{multirow}
\usepackage{hyperref}
\usepackage[ligature,reserved,inference]{semantic}
\usepackage{paralist}

\usepackage{wrapfig} %Wrap figures within text: https://tex.stackexchange.com/questions/27291/problems-with-captions-in-the-environment-wrapfigure
\usepackage{float}

\usepackage{mathtools}

\usepackage{subcaption}

\usepackage{tikz}
\usetikzlibrary{decorations.pathreplacing,calc}

\mainmatter
\title{Predictable and Performant Reactive Synthesis Modulo Theories via Functional Synthesis}

\author{Andoni Rodríguez \inst{1,2}, Felipe Gorostiaga\inst{1,3} and César Sánchez \inst{1}}
\institute{
    IMDEA Software Institute, Madrid. Spain
    \and
    Universidad Politécnica de Madrid. Spain
    \and
    CIFASIS. Argentina
}

\usepackage{graphicx}
\usepackage{amssymb} %% <- for \square and \blacksuare
\usepackage{ltlfonts}
\usepackage{listings}

\newcommand{\tupleof}[1]{\langle#1\rangle}
\newcommand{\seqof}[1]{(#1)}
\newcommand{\record}[1]{\tupleof{#1}}
\newcommand{\True}{\ensuremath{\textit{true}}\xspace}
\newcommand{\False}{\ensuremath{\textit{false}}\xspace}
\newcommand{\KWD}[1]{\ensuremath{\mathit{#1}}\xspace}
\newcommand{\AP}{\KWD{AP}}

\newcommand{\DefOR}{\ensuremath{\hspace{0.2em}\big|\hspace{0.2em}}}

\newcommand{\Always}{\LTLsquare}
\newcommand{\Event}{\LTLdiamond} 
\newcommand{\Next}{\LTLcircle}

\newcommand{\U}{\mathbin{\mathcal{U}}}

\renewcommand{\And}{\mathrel{\wedge}}
\newcommand{\Or}{\mathrel{\vee}}
\newcommand{\Impl}{\mathrel{\rightarrow}}
\newcommand{\Into}{\Impl}

\newcommand{\ltl}{\textup{LTL}}
\newcommand{\LTL}{\ensuremath{\ltl}\xspace}
\newcommand{\LTLt}{\ensuremath{\ltl^{\calT}}\xspace}

\newcommand{\phiT}{\ensuremath{\varphi^{\mathcal{T}}}\xspace}
\newcommand{\phiB}{\ensuremath{\varphi^{\mathbb{B}}}\xspace}
\newcommand{\phiExtra}{\varphi^{\textit{extra}}}
\newcommand{\phiEx}{\ensuremath{\phiExtra}\xspace}
\newcommand{\xs}{\ensuremath{\overline{x}}\xspace}
\newcommand{\ys}{\ensuremath{\overline{y}}\xspace}
\newcommand{\zs}{\ensuremath{\overline{z}}\xspace}
\newcommand{\vxs}{\ensuremath{v_{\overline{x}}}\xspace}
\newcommand{\ves}{\ensuremath{v_{\overline{e}}}\xspace}
\newcommand{\Vss}{\ensuremath{v_{\overline{s}}}\xspace}
\newcommand{\vys}{\ensuremath{v_{\overline{y}}}\xspace}
\newcommand{\vzs}{\ensuremath{v_{\overline{z}}}\xspace}

\newcommand{\Us}{\ensuremath{\overline{u}}\xspace}

\newcommand{\Ss}{\ensuremath{\overline{s}}\xspace}
\newcommand{\Es}{\ensuremath{\overline{e}}\xspace}

\newcommand{\mycal}[1]{\ensuremath{\mathcal{#1}}\xspace}
\newcommand{\calC}{\mycal{C}}

\newcommand{\calR}{\mycal{R}}
\newcommand{\calT}{\mycal{T}}

\newcommand{\VR}{\ensuremath{\textit{VR}}\xspace}

\newcommand{\ThN}{\mathcal{T}_\mathbb{N}}
\newcommand{\ThZ}{\mathcal{T}_\mathbb{Z}}
\newcommand{\ThR}{\mathcal{T}_\mathbb{R}}

\newcommand{\Vars}{\mathit{Vars}}
\newcommand{\VarsE}{\mathit{Vars}_E}
\newcommand{\VarsS}{\mathit{Vars}_S}
\newcommand{\Bool}{\mathbb{B}}

\definecolor{darkGray}{gray}{0.55} %https://tex.stackexchange.com/questions/94799/how-do-i-color-table-columns
\definecolor{lightGray}{gray}{0.85} %https://tex.stackexchange.com/questions/94799/how-do-i-color-table-columns

\newcommand{\rhoT}{\ensuremath{\rho^{\mathcal{T}}}\xspace}
\newcommand{\rhoB}{\ensuremath{\rho^{\mathbb{B}}}\xspace}

\newcommand{\Theo}{\mathcal{T}}
\newcommand{\Boolbb}{\mathbb{B}}
\newcommand{\PhiB}{\varphi^{\Boolbb}}
\newcommand{\PhiT}{\varphi^{\Theo}}
\newcommand{\phiLegal}{\varphi^\textit{legal}}

\newcommand{\dom}{\mathbb{D}}
\newcommand{\val}{\KWD{val}}

\newcommand{\es}{\overline{e}}
\newcommand{\sss}{\overline{s}}

\newcommand{\extraCons}{\Gamma}

\reservestyle{\component}{\mathsf}
\component{partitioner,provider,controller}
\newcommand{\partitioner}{\ensuremath{\<partitioner>}\xspace}
\newcommand{\provider}{\ensuremath{\<provider>}\xspace}

\newcommand{\skolem}[1]{\ensuremath{\mathtt{#1}}\xspace}

\newcommand{\skh}{\skolem{h}}

\newcommand{\TableBenchmark}{

\begin{table*}[t!]
 \centering
\begin{tabular}{|c|c||c|c|c|c|c|c||c|c|c|c|c|c|}  \hline
  Bn. & Sz. & {\textit{Prep.}} & 
  \multicolumn{2}{c|}{$1K$} & \multicolumn{3}{c|}{$10K$} & 
  \multicolumn{3}{c|}{$10K$ ($\textit{m/m.})$} & \multicolumn{3}{c|}{$10K$ ($\textit{pc.}$)}\\ 
 \cline{4-14}
  (nm.) & (vr, lt) & (s.) & Dyn. & St. & Dyn. & 
  St. & Pre. & Dyn. & St. & Pre. & Dyn. & St. & Pre. \\ [0.5ex] 
  
 \hline 
 
{\textit{Li.}} &  (5, 16) & 33.73 &  240 & 5.74 & 232 & 4.54 & $\sim 4$
 & 216 & 4.01 & $\sim 4$ & 204 & 3.52 &  $\sim 7$ \\ 
 
 \hline 
 
  {\textit{Tr.}} &  (19, 36) & 9219.11 & 272 & 5.05  & 262 & 5.03 & $\sim 9$
   & 294 & 5.08 & $\sim 14$ & 270 &  4.92& $\sim 15$  \\ 
 
 \hline 
 
 \textit{Con.}& (2, 2) & 0.09 & 104 & 2.08  & 104 & 1.89 & $\sim 2$
  & 107 & 1.94 & $\sim 2$ &  129 & 2.18 &  $\sim 2$ \\
 
  \hline 
  
 \textit{Coo.}& (3, 5) & 2.60 & 171 & 2.94 & 168 & 2.84 & $\sim 3$
   & 168 & 3.32 &  $\sim 4$ & 173 & 2.81 & $\sim 3$  \\
 
 \hline 
 
{\textit{Usb}} & (5, 8) & 346.29 & 302 & 6.04 & 304 & 4.82 & $\sim 5$
& 329 & 5.46 & $\sim 7$ & 313 & 6.00 & $\sim 6$ \\ 

 \hline 
 
  {\textit{Sta.}} &  (11, 14) & 182.1 &  295 & 4.91 & 291 & 5.19 & $\sim 6$
   & 299 & 5.24 & $\sim 9$ & 298 & 4.73 & $\sim 11$  \\[0.5em]
 
 \hline
 
\end{tabular}

\caption{Results in \cite{rodriguez24adaptive} (see Dyn.) versus ours (see St.), 
where times are measured in $\mu$.
Recall from \cite{rodriguez23boolean} that if the literals in $\varphi$ are split into clusters that do not share
variables, the Boolean abstraction process can handle each cluster
independently and composed afterwards.
}
  \label{tabBenchmark}
\end{table*}

}

\begin{document}

\maketitle

\begin{abstract}
  Reactive synthesis is the process of generating correct controllers
  from temporal logic specifications.
  Classical \LTL reactive synthesis handles (propositional) \LTL as a
  specification language.
  Boolean abstractions allow reducing \LTLt specifications (i.e.,
  \LTL with propositions replaced by literals from a theory $\calT$),
  into equi-realizable \LTL specifications.
  In this paper we extend these results into a full \emph{static}
  synthesis procedure.
  The synthesized system receives from the environment valuations of
  variables from a rich theory $\calT$ and outputs valuations of
  system variables from $\calT$.
  We use the abstraction method to synthesize a reactive Boolean controller
  from the \LTL specification, and we combine it with functional
  synthesis to obtain a static controller for the original \LTLt
  specification.
  % We combine the controllers obtained by the reactive synthesis from
  % the \LTL specification that results from the abstraction method with
  % functional synthesis, to obtain a static controller for the original
  % \LTLt specification.
  % 
  %
  We also show that our method allows \emph{adaptive
    responses} in the sense that the controller can optimize its outputs
  in order to e.g., always provide the smallest safe values.
  This is the first full static synthesis method for $\LTLt$, which is
  a deterministic program (hence predictable and efficient).
   %ABOUT ADAPTIVITY
  %We also show that our method allows \emph{adaptive
  %  responses} in the sense that the controller can receive candidate outputs from
  %an external component (e.g., an ML controller) and provide the
  %closest correct outputs with respect the Boolean controller.
  %
  %Moreover, our output generator can also consider values from the
  %history of inputs and outputs, providing smooth outputs.
  %
  %ABOUT LTLF (FINITE TRACES)
  %Our approach is applicable to both \LTL and \LTLf modulo theories.
\end{abstract}

%%% Local Variables:
%%% TeX-master: "main.tex"
%%% TeX-PDF-mode: t
%%% End:

\section{Introduction} \label{sec:intro}

Reactive synthesis for Linear Temporal Logic (\LTL)
specifications~\cite{pnueli77temporal} has received extensive research
attention~\cite{pnueli89onthesythesis}.
A specification $\varphi$ has its propositions split into those
variables controlled by the system and the rest, controlled by the
environment.
A specification is realizable if there is a strategy for the system
that produces valuations of the system variables such that all traces
generated by the controller satisfy the specification.
Realizability is the decision problem of whether such a strategy for
the system exists.
Synthesis is the process of generating one such winning strategy.
Also, both problems are decidable for \LTL~\cite{pnueli77temporal}.

A recent extension of \LTL called \LTLt (\LTL modulo theories) allows
replacing propositions with literals from a first-order theory \calT.
Given an \LTLt specification $\phiT$ an equi-realizable \LTL $\phiB$
formula can be generated, provided that the validity of \calT formulae
of the form $\exists^*\forall^*$ is
decidable~\cite{rodriguez23boolean,rodriguez24realizability}.
In \LTLt synthesis the theory variables (for example Natural o Real)
in the specification are split into environment-controlled and
system-controlled variables, and both kinds can appear in any given
literal, whereas in \LTL an atomic proposition belongs exclusively to
one player.

Note that a controller obtained from an off-the-shelf synthesis
procedure for the Booleanized LTL formula $\phiB$ cannot be directly
used as a controller for $\phiT$, because it must handle rich input
and output values.
Previous similar approaches either (1)
focus on the satisfiability problem and not in realizability 
(e.g., \cite{geatti22linear,geatti23decidable});
or (2) cannot be adapted to arbitrary $\calT$
\cite{katis2018validity} or (3) do not guarantee termination
\cite{maderbacher2021reactive,samuel23symbolic,heim24solving}, which makes these
solutions incomplete.
Recently, \cite{rodriguez24adaptive} presented a method for synthesis
of a fragment of decidable $\LTLt$ specifications, which relies on the
use of SMT solvers on-the-fly at every reaction, so it does not
produce a standalone controller.
This precludes the application to real embedded systems where
controllers frequently operate because (1) the SMT solver is not
guaranteed to terminate (particularly with limited resources), (2) the
solver may not return the same values provided the same formula
(affecting \emph{predictability}) and (3) invoking solvers on the fly
has an impact on the performance and the reaction time.

In this paper we present a \emph{static} synthesis procedure for \LTLt
specifications.
Our method proceeds as follows.
We first obtain a Boolean controller $C$ for the equi-realizable
Boolean specification $\phiB$.
The controller for $\phiT$ uses two additional components: a
\partitioner, that transforms the environment input $u$ into the
corresponding Boolean input to $C$ and a \provider that receives the
input $u$ and the reaction from $C$, and produces the reaction $v$.
Our \provider, instead of performing SMT calls, is implemented as a a
collection of Skolem functions $f$ that generate, given $u$, an output
$v$, such that the literals in $\phiT$ agree with the reaction chosen
by $C$.
%
% The main property of the resulting controller is that the sequence of
% rich input $u$ and output valuations $v$ correspond to precisely the
% sequence of Boolean inputs and outputs that $C$ generates.
%
Since each trace produced by our controller satisfies $\phiT$, the
composition of the \partitioner, the Boolean controller $C$ and the
\provider is a controller for $\phiT$
Therefore, the procedure described here is a static synthesis
procedure for specifications in $\LTLt$.
The resulting controller is standalone deterministic program, which is
predictable and highly performant.

Next, we exploit the fact that we can synthesize Skolem functions that
additionally receive a set of constraints to not only generate outputs
that satisfy the desired literals, but that also optimize certain
criteria from the set of possible correct outputs (e.g., to provide
the smallest value among possible values).
We call this technique \textit{adaptivity}.
%
%ABOUT ADAPTIVITY (NOT IN THIS PAPER)
%We also show how to use Skolem functions that receive additional
%parameters, including previous values of inputs and outputs, which
%allows creating realistic smooth controllers.
%
%Finally, we show how to connect the output of an external component
%that suggests output values as an additional input of the Skolem
%function, which then produces the closest safe value to the value
%proposed.
%
%The idea is that the external component can be a sophisticated but
%potentially unsafe controller, for example, a controller trained using
%machine learning.
%
%Our ``neuro-symbolic'' approach produces a safe controller in which
%the outputs of a certified controller are safe and guided by a neural
%model. %(which is dual to the notion of shielding
%\cite{alshiekhETAL2017safeReinforcementLearningShielding}).
%
All these results are applicable to $\LTLt$ on infinite or on finite
traces, using appropriate synthesis tools for the resulting $\phiB$.

In summary, the contributions of this paper are:
(1) a formalization and soundness proof of the controller architecture for synthesis for $\LTLt$;
(2) a derived correct method to synthetise static controllers from $\LTLt$ specifications combining Boolean abstraction, reactive synthesis and functional synthesis;
(3) a formalization of the limits and capabilities of using Skolem functions, 
showing their power to model adaptivity;
%ABOUT ADAPTIVITY
%(3) a theory-specific analysis that shows we can construct controllers optimal to 
%criteria like temporal smoothness and value proximity; and
(4) an extensive empirical evaluation that shows our method predictable and fast.
To the best of our knowledge, this is the first full static reactive synthesis
approach for $\LTLt$ specifications.
Moreover, since our approach leverages off-the-shelf components
(reactive synthesis, functional synthesis), it would immediately
benefit from advances in those areas and also from discoveries in 
decidable fragments of $\LTLt$ realizability.
The remainder of the paper is structured as follows.
Sec.~\ref{sec:prelim} contains preliminary definitions, including the
Boolean abstraction method from \cite{rodriguez23boolean} and a
running example that is used in the rest of the paper.
Sec.~\ref{sec:static} formalizes the
controller architecture and proves its correctness.
Sec.~\ref{sec:adapt} introduces an adaptive extension of our approach.
%ABOUT ADAPTIVITY
%adaptivity for
%enriched $\LTLt$, the neurosymbolic approach and the computation of
%adaptive Skolem functions.
%
Sec.~\ref{sec:empirical} contains an empirical evaluation. % and finally,
Sec.~\ref{sec:conclusion} shows related work and concludes.

%%% Local Variables:
%%% TeX-master: "main.tex"
%%% TeX-PDF-mode: t
%%% End:

%\newpage
\section{Preliminaries} \label{sec:prelim}

\subsubsection{First-order Theories.}
In this paper we use first-order theories.
We describe theories with single domain for simplicity, but this can
be easily extended to multiple sorts.
A first-order theory $\calT$ (see e.g.,~\cite{bradley07calculus}) is
described by a signature $\Sigma$, which consists of a finite set of
functions and constants, a set of variables and a domain.
The domain $\dom$ of a theory $\calT$ is the sort of its variables.
For example, the domain of non-linear real arithmetic $\ThR$ is
$\mathbb{R}$ and we denote this by $\dom(\ThR)=\mathbb{R}$ or simply
by $\dom$ if it is clear from the context.
A formula $\varphi$ is valid in $\calT$ if, for every interpretation
$I$ of $\calT$, then $I \vDash \varphi$.
A fragment of a theory $\calT$ is a syntactically-restricted subset of
formulae of $\calT$.
Given a formula $\psi$, we use $\psi[\xs \leftarrow \Us]$ for the
substitution of variables $\xs$ by terms $\Us$ (typically constants).

\subsubsection{Reactive Synthesis.}
We fix a finite set of atomic propositions $\AP$.
Then, $\Sigma=2^\AP$ is the alphabet of valuations, and $\Sigma^*$ and
$\Sigma^\omega$ are the set of finite and infinite traces
respectively.
Given a trace $\sigma$ we use $\sigma(i)$ for the letter at position
$i$ and $\sigma^i$ for the suffix trace that starts at position $i$.
The syntax of propositional
\LTL~\cite{pnueli77temporal,manna95temporal} is:
\[
  \varphi  ::= \top \;\DefOR\; a \;\DefOR\; \varphi \lor \varphi \;\DefOR\; \neg \varphi
  \;\DefOR\; \Next \varphi \;\DefOR\; \varphi \U\varphi
\]
where $a\in \AP$; $\lor$, $\land$ and $\neg$ are the usual Boolean
disjunction, conjunction and negation; and $\Next$ and $\U$ are the
next and until temporal operators.
The semantics of \LTL associates traces $\sigma\in\Sigma^\omega$ with
\LTL fomulae as follows:
\[
  \begin{array}{l@{\hspace{0.3em}}c@{\hspace{0.3em}}l}
    \sigma \models \top && \text{always holds} \\
    \sigma \models a & \text{iff } & a \in\sigma(0) \\
    \sigma \models \varphi_1 \Or \varphi_2 & \text{iff } & \sigma\models \varphi_1 \text{ or } \sigma\models \varphi_2 \\
     % \sigma \models \varphi_1 \And \varphi_2 & \text{iff } & \sigma\models \varphi_1 \text{ and } \sigma\models \varphi_2 \\
     \sigma \models \neg \varphi & \text{iff } & \sigma \not\models\varphi \\
     \sigma \models \Next \varphi & \text{iff } & \sigma^1\models \varphi \\
     \sigma \models \varphi_1 \U \varphi_2 & \text{iff } & \text{for some } i\geq 0\;\; \sigma^i\models \varphi_2, \text{ and } \text{for all } 0\leq j<i, \sigma^j\models\varphi_1 \\
  \end{array}
\]
We use common derived operators like $\vee$, $\calR$, $\Event$ and
$\Always$.
% 
%ABOUT LTLF (FINITE TRACES)
%The variant \LTLf of \LTL for finite traces
%\cite{manna95temporal,giacomo2013finiteTraces} borrows the syntax from
%LTL, with an additional next operator $\WNext$.
% 
%The semantics are adapted for finite words $\sigma\in\Sigma^*$ as
%follows (all other operators remain the same):
%\[
%  \begin{array}{l@{\hspace{0.3em}}c@{\hspace{0.3em}}l@{\hspace{19em}}}
%    \sigma \models \Next \varphi & \text{iff } & |\sigma|>0 \text{ and } \sigma^1\models \varphi \\
%    \sigma \models \WNext \varphi & \text{iff} & |\sigma|=0 \text{ or } \sigma^1\models \varphi \\
%  \end{array} 
%\]
% 
Reactive synthesis
\cite{pnueli89onthesythesis,pnueli89onthesythesis:b,bloem12synthesis,finkbeiner16synthesis}
is the problem of automatically constructing a system based on an \LTL
specification $\varphi$, where the atomic propositions of $\varphi$
(\AP) are divided into propositions $\Es=\VarsE(\varphi)$ controlled
by the environment and $\Ss=\VarsS(\varphi)$ controlled by the system
(with $\Es\cup\Ss=\AP$ and $\Es\cap\Ss=\emptyset$).
A reactive specification corresponds to a turn-based game where the
environment and system players alternate.
In each turn, the environment produces values for $\Es$, and the
system responds with values for $\Ss$.
A valuation is a map from $\Es$ into $\Bool$ (similarly for $\Ss$).
We use $\val(\Es)$ and $\val(\Ss)$ for valuations.
A play is an infinite sequence of turns 
%ABOUT LTLF (FINITE TRACES)
%(finite for \LTLf) 
and induces
a trace $\sigma$ by joining at each position the valuations that the
environment and system players choose.
The system player wins a play if the trace satisfies $\varphi$.
A strategy for the system is a tuple $\rho: \tupleof{Q,q_0,\delta,o}$
where $Q$ is a finite set of states, $q_0\in Q$ is the inital state,
$\delta:Q\times \val(\Es) \Into Q$ is the transition function and
$o:Q\times\val(\Es)\Into\val(\Ss)$ is the output function.
A play $\seqof{(\Es_0,\Ss_0),(\Es_1,\Ss_1),\ldots}$ is played according to $\rho$
if the sequence $\seqof{(\Es_0,\Ss_0,q_0),(\Es_1,\Ss_1,q_1),\ldots}$ satisfies
that $q_{i+1}=\delta(q_i,\Es_i)$ and $\Ss_i=o(q_i,\Es_i)$ for all $i
\geq 0$.
A strategy $\rho$ is wining for the system if all plays played
according to $\rho$ satisfy $\varphi$.
We will use \emph{strategy} and \emph{controller} interchangeably.

\subsubsection{Linear Temporal Logic Modulo Theories. }

The syntax of \LTLt replaces atoms $a$ by literals $l$ from theory
$\calT$.
We use $\Vars(l)$ for the variables in literal $l$ and
$\Vars(\varphi)$ for the union of the variables that occur in the
literals of $\varphi$.
A valuation for a set of vars $\zs$ is a map from $\zs$ into $\dom$.
The alphabet of a formula $\varphi$ is
$\Sigma_{\calT}: \Vars(\varphi) \Into \dom$.
The semantics of \LTLt associate traces
$\sigma\in\Sigma_{\calT}^\omega$ with formulae, where for atomic
propositions $\sigma \models l$ holds iff
$\sigma(0) \vDash_{\calT} l$, that is, if the valuation $\sigma(0)$
makes the literal $l$ true.
The rest of the operators are as in \LTL.

For realizability and synthesis from \LTLt, the variables in
$\Vars(\varphi)$ are split into those variables controlled by the
environment ($\xs$ or $\VarsE(\varphi)$) and those controlled by the
system ($\ys$ or $\VarsE(\varphi)$).
We use $\varphi(\xs,\ys)$ to denote that $\xs\cup\ys$ are the
variables occurring in $\varphi$ (where $\xs \cup \ys=\Vars(\varphi)$
and $\xs \cap \ys = \emptyset$).
%
% The alphabet $\Sigma_{\calT}$ is a valuation of the variables in
% $\xs\cup\ys$.
%
A trace is an infinite sequence of valuations of $\xs$ and $\ys$,
which induces an infinite sequence of Boolean values for each of the
literals at each position, and ultimately a valuation of $\varphi$.
For instance, given $\psi = \Always(y>x)$ the trace
$\seqof{\record{x:4,y:5},\record{x:9,y:7},\ldots}$ induces
$\seqof{\record{l:\True},\record{l:\False},\ldots}$ for the literal $l
= (y>x)$.
An \LTLt specification corresponds to a game with an infinite arena,
where positions can have infinitely many successors.
A strategy now for the system is a tuple
$\rhoT: \tupleof{Q,q_0,\delta,o}$ where $Q$ and $q_0$ are as before and
$\delta:Q\times \val(\xs) \Into Q$ is the transition function and
$o:Q\times\val(\xs)\Into\val(\ys)$ is the output function.

\subsubsection{Boolean Abstraction.}
The Boolean abstraction method~\cite{rodriguez23boolean} transforms an
\LTLt specification $\phiT$ into an equi-realizable \LTL specification
$\phiB$.
The resulting \LTL formula can be passed to an off-the-shelf synthesis
engine, which generates a controller for realizable specifications.
The process of Boolean abstraction involves transforming an input
formula $\phiT$, which contains literals $l_i$, into a new
specification $\phiB = \phiT[l_i \leftarrow s_i] \wedge \phiEx$, where
$\Ss=\{s_i|\text{for each } l_i\}$ is a set of fresh atomic propositions
controlled by the system---such that $s_i$ replaces $l_i$---and where
$\phiEx$ is an additional sub-formula that captures the dependencies
between the $\Ss$ variables\footnote{The Boolean abstraction process can
  substitute larger sub-formulae than literals (as long as they do not
  contain temporal operators).}.
The formula $\phiEx$ also includes additional environment variables
$\es$ (controlled by the environment) that encode the power of the
environment to leave the system with the power to choose certain
valuations of the variables $\Ss$.
The formula $\phiEx$ also constraints the environment in such a way
that exactly one of the variables in $\es$ is true.
 
A \emph{choice} $c$ is a valuation of $\Ss$,
$c(s_i)=\True$ means that $s_i$ is in the choice.
We write $s_i\in c$ as a synonym of $c(s_i)=\True$.
The characteristic formula $f_c(\xs,\ys)$ of a choice $c$ is
\(
  f_c=\bigwedge_{s_i\in c} l_i \And \bigwedge_{s_i\notin c}\neg l_i.
\)
Note that we often represent choices as valuation $\Vss$ of the Boolean variables $\Ss$ 
(which map each variable in $\sss$ to $\True$ or $\False$).
We use $\calC$ for the set of choices (that is, the set of sets of $\Ss$).
A \emph{reaction} $r\subset\calC$ is a set of choices, which
characterizes the possible responses of the system as the result of a
move by the environment.
The \emph{characteristic formula} $f_r(\xs)$ of a reaction $r$ is:
\[
  (\bigwedge_{c\in r} \exists\ys. f_c) \And (\bigwedge_{c\notin r} \forall\ys\neg f_c)
\]
A reaction $r$ is valid whenever $\exists \xs.f_r(\xs)$ is valid.

Intuitively, $f_r$ states that for some valuations of the variables
$\xs$ controlled by the environment, the system can respond with
valuations of $\ys$ making the literals in some choice $c\in r$ but
cannot respond with valuations making the literals in choices
$c\notin r$.
The set of valid reactions partitions precisely the moves of the
environment in terms of the reaction power left to the system.
For each valid reaction $r$ there is a fresh environment variable
$e\in\es$.
Hence, the restriction in $\phiEx$ that forces the environment to make
exactly one variable in $\es$ true corresponds to the environment
choosing a reaction $r$ (when the corresponding $e$ is true).

Boolean abstraction~\cite{rodriguez23boolean} uses the set of valid
reactions to produce an equi-realizable $\phiB$ from a formula
$\phiT$, which are in the same temporal fragment.

% The main result in~\cite{rodriguez23boolean} is that $\phiT$ is
% realizable if and only if $\phiB$ is realizable, but no synthesis
% procedure was given for $\phiT$.
% %
% The abstraction algorithm is agnostic to the temporal fragment of
% $\phiT$, which implies that $\phiB$ preserves the fragment of $\phiT$
% (e.g., from a safety $\phiT$ it computes a safety $\phiB$).

%\subsubsection{Motivating running example.}

\begin{example} [Running example] \label{exRunning}
  Let $\phiT(\xs,\ys)$ be the following specification (where
  $\xs=\{x\}$ is controlled by the environment and $\ys=\{y\}$ by the
  system):
  \[ 
    \phiT =\square \big[ \big((x<2) \Into \Next(y>1)\big)  \And
    \big((x \geq  2) \shortrightarrow (y \leq x)\big)\big].
  \]
In theory $\ThZ$ this specification is realizable (consider the
strategy to always play $y=2$).
In this theory, the Boolean abstraction first introduces $s_0$ to
abstract $(x<2)$, $s_1$ to abstract $(y>1)$ and $s_2$ to abstract
$(y \leq x)$.
Then %
\( \phiB = \varphi'' \wedge \Always (\phiLegal \Into \phiExtra) \)
where
$\varphi'' = (s_0 \shortrightarrow \Next s_1) \wedge (\neg s_0
\shortrightarrow s_2)$ is a direct abstraction of $\phiT$.
Finally, $\phiExtra$ captures the depenencies between the abstracted
variables:
\newcommand{\NN}{\phantom{\neg}}
\begin{align*}
  \phiExtra: &
  \begin{pmatrix}
    \begin{array}{lrcl}
      \phantom{\wedge}& \big(e_0 & \Into & \big( [\NN s_0 \wedge s_1 \wedge \neg s_2] \vee [\NN s_0 \wedge \neg s_1 \wedge s_2] \big) \\[0.27em]
      % &\wedge& \\
      %\wedge & \big(e_{0^+} & \Into & \big( [\NN s_0 \wedge s_1 \wedge \neg s_2] \vee [\NN s_0 \wedge \neg s_1 \wedge \NN s_2] \vee [\NN s_0 \wedge \neg s_1 \wedge \NN s_2] \big) \\[0.27em]
      \wedge & \big(e_1 & \Into & \big( [\neg s_0 \wedge s_1 \wedge \neg s_2] \vee [\neg s_0 \wedge \neg s_1 \wedge \NN s_2] \vee [\neg s_0 \wedge \NN s_1 \wedge \neg s_2] \big)
    \end{array}
  \end{pmatrix} %\\
  %\phiLegal: & (e_0 \vee e_1 \vee e_2) \wedge (e_0 \Into \neg (e_1\wedge e_2)) \wedge
%(e_1 \Into \neg (e_0\wedge e_2)) \wedge (e_2 \Into \neg(e_0\wedge e_1))
\end{align*}
and $\phiLegal: (e_0 \vee e_1) \wedge (e_0 \leftrightarrow \neg e_1)$,
where $\es=\{e_0, e_1 \}$ belong to the environment and
represent $(x<2)$ and $(x \geq 2)$, respectively.
Thus, $\phiLegal$ encodes that $e_0$ and $e_1$ characterize
a partition of the (infinite) input valuations of the environment (and
that precisely one of $\es$ are true in every move).
For example, the valuation
$v_e = \record{e_0:\False,e_1:\True}$ of $\es$
corresponds to the choice of the environment where only $e_1$ is true.
%
  %Note that a clever environment will never play this valuation, since
  %$\tupleof{e_0:\True,e_{0+}:\False,e_1:\False}$ offers strictly less
  %power to the system, so the latter is a strictly better move.
%
%Therefore, an abstraction that ignores $e_{0^+}$ is also
%equi-realizable to $\phiT$, so we will consider the shorter version
%for simplicity in the presentation.
%
Sub-formulae like $(s_0 \wedge s_1 \wedge \neg s_2)$ represent the
\textit{choices} of the system (in this case, $c=\{s_0,s_1\}$), that is, given a decision of the
environment (a valuation of $\es$ that makes exactly one variable $e$ true),
the system can \emph{react} with one of the choices $c$ in the
disjunction implied by $e$.
We denote $c_0 = \{s_0,s_1,s_2\}$, $c_1 = \{s_0,s_1\}$, $c_2 = \{s_0,s_2\}$,
$c_3 = \{s_0\}$, $c_4 = \{s_1,s_2\}$, $c_5 = \{s_1\}$,
$c_6 = \{s_2\}$ and $c_7 = \emptyset$.
Note that e.g., $c_1$ can be represented as $\Vss = \tupleof{s_0 : \True, s_1 : \True, s_2 : \False}$.
\end{example}

%%% Local Variables:
%%% TeX-master: "main.tex"
%%% TeX-PDF-mode: t
%%% End:

\section{Static Reactive Synthesis Modulo Theories}
\label{sec:static}

The Boolean abstraction method~\cite{rodriguez23boolean} reduces an
\LTLt formula \phiT into an equi-realizable \LTL specification \phiB,
but it does not present a synthesis procedure for \phiT.
We solve this problem here by providing a full static synthesis method
for \LTLt.
Our procedure builds a controller for realizable $\phiT$
specifications that handles inputs and outputs from a rich domain
$\dom(\calT)$, using as a building block the synthetized Boolean
controller for $\phiB$ 
and other two sub-components.

\subsection{Formal Architecture} \label{subsec:architecture}

\begin{figure}[b!]
%\begin{minipage}{\textwidth}
\centering
  \includegraphics[width=\linewidth]{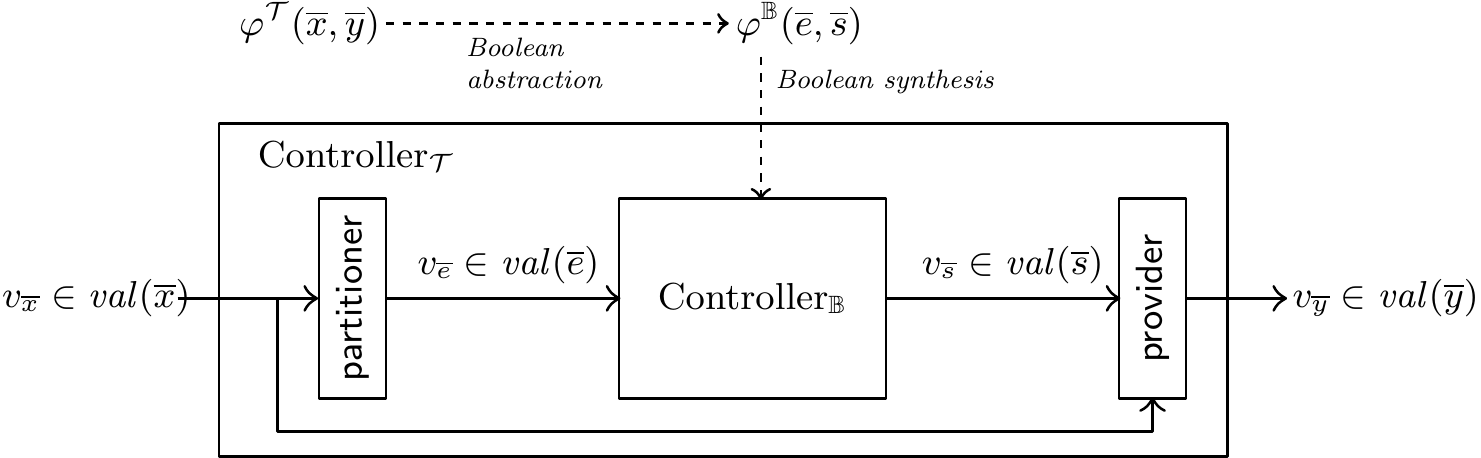}
  \caption{A controller architecture for reactive synthesis of \LTLt specifications.}
  \label{figArchitecture}
%\end{minipage}%
% \\
% \begin{minipage}{\textwidth}
% \centering
%   \includegraphics[width=.8\linewidth]{figures/clusters.pdf}
%   \caption{Architecture with clusters.}
%   \label{figArchClusters}
% \end{minipage}
\end{figure}
%\felipe{(Fig. 1) In the text we use $\ves\in\val(\Es)$ instead of
%$\ves\in\val(\Es)$. I think we should unify the notation, and reserve
%the colon for functions}

We call our approach static \LTLt synthesis (see
Fig.~\ref{figArchitecture}).
Our method starts from $\phiT(\xs,\ys)$ and statically generates a
Boolean controller $\rhoB$ for $\PhiB$ and combines it with a
\partitioner and a \provider 
(generated from the abstraction process of $\phiT$ to $\phiB$)
handle the inputs and outputs from $\dom$.
%
% We start from a formula $\phiT(\xs,\ys)$ with variables $\xs$
% controlled by the environment and variables $\ys$ controlled by the
% system.
%
At run-time, at each instant the resulting controller follows these
steps:
\begin{compactenum}[(1)]
\item a valuation $\vxs \in \val(\xs)$ is provided by the environment;
\item the \partitioner discretizes $\vxs$ generating a Boolean
  valuation $\ves\in\val(\Es)$ of input variables for $\rhoB$.
\item $\rhoB$ responds with a valuation $\Vss\in\val(\Ss)$ of the
  variables $\Ss$ that $\rhoB$ controls.
\item the \provider produces a valuation $\vys\in\val(\ys)$ of the output
  variables that together with $\vxs$ guarantee that the literals from
  $\phiT$ will be evaluated as indicated by the choice $c$ that
  corresponds to $\Vss$ indicated by $\rhoB$.
  This step corresponds to finding a model of
  $\exists \ys .f_c([\xs \leftarrow \vxs],\ys)$.
\end{compactenum}  
For step (4) one approach is to invoke an SMT solver on the fly to
generate models (proper values of $\vys$), which is guaranteed to be
satisfiable (by the soundness of the Boolean abstraction method).
However, most uses of controllers cannot use SMT solvers dynamically.
Moreover, note that the formula to be solved has quantifier
alternations (within $f_c$) which is currently challenging for
state-of-the art SMT solving technology for many theories.
In this paper we present an alternative: a method that produces a
totally static controller, using Skolem functions associated to
each $(e, c)$ pair.
These Skolem functions are models of the formula
\[
  \forall \xs. \exists \ys. f_r(\xs) \Into f_c(\xs,\ys)
\]
Recall that $f_r(\xs)$ is the formula that characterizes the
environment valuations for which $r$ captures the possible responses
after receiving $\xs$, according to reasoning in the theory $\calT$.
%
% In Sec.~\ref{sec:adapt} we show how to refine the \provider to
% create \emph{adaptive functions}.

% We now formalize and proof correctness of the approach, for which we
% will denote with $n$ number of the environment variables in $\phiT$,
% $u$ the number of environment variables in $\phiB$, $g$ the number of
% system variables in $\phiB$ and with $m$ the number of system
% variables in $\phiT$.
% %
% \cesar{Move these definitions to Preliminaries}
% %
% We use $\xs \in E_{\calT}^n$ for environment variables from $\phiT$
% and $\ys \in S_{\calT}^m$ for the system variables (both of type
% $\dom$.

% $\dom$, $\es\in E_{\mathbb{B}}^u$ the corresponding Boolean variables,
% $\sss \in S_{\mathbb{B}}^g$ the Boolean response output variables of
% the system; and 
% %
% Last, we denote with $\vs \in \val(\xs)$ valuations of the environment for $\xs$,
% with $e_k \in \val(\es)$ valuations of the environment for $\es$,
% with $c_i \in \val(\sss)$ valuations of the system for $\sss$,
% with $\ws \in \val(\ys)$ valuations of the system for $\ys$.

%
%Figures side by side: https://tex.stackexchange.com/questions/37581/latex-figures-side-by-side

\newcommand{\ek}{\ensuremath{\overline{e}_k}\xspace}

\subsubsection{Partitioner.}

At each timestep, the partitioner receives a valuation
$\vxs\in\val(\xs)$ of the environment variables and finds the input
variable $e_k$ to be fed to the Boolean controller.
The partitioner must find the entry $(e,r)$ in the table of valid
reactions for which $f_r(\vxs)$ is valid and return $e$.
For instance, recall partitions $e_0$ and $e_1$ from
Ex.~\ref{exRunning}, then an input trace $\seqof{\record{x:4}, \record{x:4},
\record{x:1}, \record{x:0}, \record{x:2},\ldots}$ will be partitioned into
$\seqof{\record{e:e_1},\record{e:e_1},\record{e:e_0},\record{e:e_0},\record{e:e_1},\ldots}$
(for simplicity, here we show the only Boolean variable $e_i$ that is
true).
The following defines a legal partitioner.

% The partition obeys to a set in $\dom(\calT)$ in which $\xs$ is
% located (e.g., $x$ such that $(x<2)$) and is defined at its most
% general form using reactions $r$, where $\xs$ whether verified or
% falsifies a formula $r(\xs)$ with the shape
% $\exists \xs. [Q_0 \ys. (\xs, \ys) \wedge Q_1 \ys. (\xs, \ys) \wedge
% ... \wedge Q_{n-1} \ys. (\xs, \ys)]$, where
% $Q_0, Q_1, ..., Q_{n-1} \in Q=\{\exists, \neg \exists\}$ and $n$ is
% exponential in the size of choices.
% %
% Recall from Sec.~\ref{sec:prelim} that the Booleanization algorithm essentially computes the set 
% $\VR$ of valid reactions, which characterizes (in a discrete manner) every decision $e_k$ that the environment can take 
% and every choice $c_i$ with which the system can respond to each $e_k$.
% %
% Also, recall that each reaction $r \in \VR $ is associated with a partition identified by $e_k$. 
% %
% Moreover, consider the set $\mathcal{D} = \{e_r | r\in\VR\}$ that contains the Boolean
% decision variables for the environment, one for each reaction and where for all $e_r \in \mathcal{D}$, then $e_r \in \val(\es)$.
% %
% We use the function $\text{dec}: \VR \sImplies \mathcal{D}$, which
% receives a reaction $r$ and returns its decision variable, and its
% inverse function $\text{dec}^{-1}: \mathcal{D} \sImplies \VR$, which takes a
% decision variable $e_r$ and returns its reaction $r$.
% %
% It is easy to see that these functions exist.
% %
% Then:
%
\begin{definition}
  \label{def:partitioner} \label{defPartitioner}
  Let $\phiT(\xs,\ys)$ be an \LTLt specification and $\phiB(\Es,\Ss)$
  its Boolean abstraction.
  A \partitioner is a function $\alpha: \val(\xs) \Into \Es$ such
  that if $(e,r)$ is a valid reaction and $f_r[\xs\leftarrow\vxs]$
  is valid, then $\alpha(\vxs)=e$. 
\end{definition}

\noindent{}Note that, by the soundness of the Boolean abstraction method, there
is one
\begin{wrapfigure}[6]{l}{0.50\textwidth} % Adjust the number of lines and width as needed
  \vspace{-2.5em}
  \begin{minipage}{0.50\textwidth}
    \begin{algorithm}[H]
 \SetAlgorithmName{Alg.}{} 
 %\LinesNumbered
 \textbf{Input: }$\vxs\in\val(\xs)$ \\
  \ForAll{$(e,r)\in \VR(\varphi)$}
  {
    \If{$f_{r}[\xs\leftarrow\vxs]$ is valid}
    {
     \Return $e$ \\
     }
  }
 % \Return \textit{Error}
  %\rememberlines
  % \caption{A brute force \partitioner $\alpha$.}
  \caption{\mbox{A brute force partitioner $\alpha$.}}
  \label{algoPartitioner}
    \end{algorithm}
  \end{minipage}
\end{wrapfigure}
and only one such candidate $e$ for every input $\vxs$.
Note that $\alpha(\vxs)=e$ induces a valuation $\ves$ of the variables
$\Es$ by $\vxs(e)=\True$ and $\vxs(e')=\False$ for all other
$e'\neq e$.
Alg.~\ref{algoPartitioner} shows a brute force method to find variable $e$.

\subsubsection{Controller.}

The Boolean $\<controller>$ receives the discrete environment input
$\ves$ and produces a discrete output $\Vss$ that represents the
selected choices according to a winning strategy for $\phiB$.
This controller $\rhoB$ can obtained using off-the-shelf reactive
synthesis tools.
This controller $\rhoB$ produces a valuation $\Vss\in\val(\Ss)$ at
every step, guaranteeing that the trace produced satisfies $\phiB$.

Consider an instant where only $e$ is true in the input $\ves$,
and let $r$ be the valid reaction corresponding to $e$; then
if $\Vss$ is the output produced by $\rhoB$ from $\ves$ the choice
$c:\{s_i|\Vss(s_i)=\True\}$ belongs to $c\in r$.
This is forced by the $\phiExtra$ constraint in the construction of
$\phiB$ from $\phiT$ in the Boolean abstraction method.
To better illustrate this, recall Ex.~\ref{exRunning} and among the
possible winning strategies that the system has, consider the
following.
If $\ves$ is $\record{e_0:\True,e_1:\False}$, then the output choice
$\Vss$ is $\record{s_0:\True, s_1:\True, s_2:\False}$ (i.e., $c_1$).
On the other hand, if $\ves(e_1)$ is $\record{e_0:\False,e_1:\True}$
then output choice $\Vss$ is
$\record{s_0:\False, s_1:\True, s_2:\True}$ (i.e., $c_4$).
%
% Then, an input trace where
% $(\tupleof{x:4},\tupleof{x:4},\tupleof{x:1},\tupleof{x:0},\tupleof{x:2},\ldots)$
% is partitioned into
% $(\tupleof{e_0:\False,e_1:\True},e_1,e_0,e_0,e_1,\ldots)$ would
% produce Boolean $(c_4,c_4,c_1,c_1,c_4,\ldots)$.
%
% We now describe the last step from $\mathbb{B}$ to $\dom(\calT)$.
% \[
%   \begin{array}{llllll}
%     \tupleof{x:4} & \tupleof{x:4} & \tupleof{x:1} & \tupleof{x:0} & \tupleof{x:2} & \ldots \\
%     \tupleof{e_0:\False,e_1:\True} &
%     \tupleof{e_0:\False,e_1:\True} &
%     \tupleof{e_0:\True,e_1:\False} &
%     \tupleof{e_0:\True,e_1:\False} &
%     \tupleof{e_0:\False,e_1:\True} & \ldots\\
%     \tupleof{s_0:\False,s_1:\True,s_2:\True} & 
%     \tupleof{s_0:\False,s_1:\True,s_2:\True} & 
%     \tupleof{s_0:\Treu,s_1:\True,s_2:\False} & 
%     \tupleof{s_0:\Treu,s_1:\True,s_2:\False} & 
%     \tupleof{s_0:\False,s_1:\True,s_2:\True} &  \ldots
%   \end{array}
% \]    

\subsubsection{Provider.}

The discrete behavior of the Boolean controller requires an
additional component that produces a valuation $\vys\in\val(\ys)$ of
the system variables over $\ys$ satisfying $\phiT$.
The \provider receives the choice and the input $\vxs\in\val(\xs)$, and
substitutes $\vxs$ for $\xs$ in $f_c$:
\( f_c([\xs \leftarrow\vxs],\ys) \).
The goal of the provider is to find a proper valuation for $\ys$.

\begin{definition}[Provider] \label{defProvider} A \provider is a
  function $\beta: \val(\xs)\times \val(\Ss) \Into \val(\ys)$ such
  that for every $\vxs\in\val(\xs)$ and choice $c\in\val(\Ss)$ , the
  following holds
  \[ f_c(\xs\leftarrow\vxs, \ys\leftarrow \beta(\vxs,c)). \]
\end{definition}

\noindent We will show below that if $\vxs$ is an input to a
\partitioner, $r$ is the valid corresponding reaction, and $c$ is one
of the winning choices of the controller (that is, $c\in r$), then the
following formula is valid.
\[
  [\exists \ys. f_c(\ys,\xs\leftarrow\vxs)]
\]  
This formula can be discharged into a solver with capabilities to
produce a model $\vys$ (e.g., an SMT solver like Z3
\cite{demoura08z3}), 
which is exactly the \textbf{dynamic} approach presented at \cite{rodriguez24adaptive}.
\begin{example}
  \label{ex:provider}
  Consider again Ex.~\ref{exRunning} and input trace
  $\seqof{\record{x:4}, \record{x:4}, \record{x:1}, \record{x:0},
  \record{x:2},\ldots}$.
  This trace is mapped into the following discrete input trace
  $\seqof{\record{c:c_4},\record{c:c_4},\record{c:c_1},\record{c:c_1},\record{c:c_4},\ldots}$.
  Recall that
  $s_0$ abstracts $(x<2)$, $s_1$ abstracts $(y>1)$ and $s_2$ abstracts
  $(y \leq x)$.
  Then, the output trace must be a sequence $\vys$ of values of $y$
  such that the following holds:
\[ \begin{array}{r@{}ll}
     (       & [\neg (4<2) \wedge (y>1) \wedge \phantom{\neg}(y \leq 4)],\\
             & [\neg (4<2) \wedge (y>1) \wedge \phantom{\neg}(y \leq 4)],\\
             & [\phantom{\neg}(1<2) \wedge (y>1) \wedge \neg (y \leq 1)],\\
             & [\phantom{\neg}(0<2) \wedge (y>1) \wedge \neg (y \leq 0)],\\
             & [\neg (2<2) \wedge (y>1) \wedge \phantom{\neg}(y \leq 2)],\ldots)%
   \end{array}
 \]
 One such possible sequence is $\seqof{\record{y:2},\record{y:2},\record{y:2},\record{y:2},\record{y:2},\ldots}$.
% 
%
%  {\begin{small}
% \[ \begin{array}{r@{}l@{\hspace{0.3em}}l}
%      \langle & [\neg (4<2) \wedge (y>1) \wedge (y \leq 4)], & [\neg (4<2) \wedge (y>1) \wedge (y \leq 4)], \\
%      & [(1<2) \wedge (y>1) \wedge \neg (y \leq 1)], & [(0<2) \wedge (y>1) \wedge \neg (y \leq 0)],\\
%              & [\neg (2<2) \wedge (y>1) \wedge (y \leq 2)], & \ldots\rangle%
%    \end{array}
%  \]
%  }
%
% $\{ [\neg (4<2) \wedge (y_0>1) \wedge (y_0 \leq 4)],%
% [\neg (4<2) \wedge (y_1>1) \wedge (y_1 \leq 4)],%
% [(1<2) \wedge (y_2>1) \wedge \neg (y_2 \leq 1)],%
% [(0<2) \wedge (y_3>1) \wedge \neg (y_3 \leq 0)],%
% [\neg (2<2) \wedge (y_4>1) \wedge (y_4 \leq 2)]%
% \}$.
%
Note how $\xs$ is replaced in each timestep by concrete input $\vxs$.
However, many different values $\vys$ exist that satisfy the output
trace (e.g.
$\seqof{\record{y:2},\record{y:3},\record{y:3},\record{y:4},\record{y:2},\ldots}$.)
\end{example}

\subsubsection{Correctness.}
The Boolean system strategy $\rhoB:\tupleof{Q,q_0,\delta,o}$ for
$\phiB$ produces, at every timestep, a valuation of $\Ss$ from a valuation
of $\Es$.
We now define a strategy $\rhoT$ of the system in $\phiT$ and prove
that all moves played according to $\rhoT$ are winning for the system;
i.e., all produced traces satisfy $\phiT$.
Intuitively, $\rhoT$ composes the \partitioner, which
translates inputs to the Boolean controller, collects the move chosen
by the Boolean controller and then uses the \provider to generate an
output.

\begin{definition}[Combined Strategy]
  \label{def:combined}
Given a \partitioner $\alpha$, a controller $\rhoB$ for $\phiB$
and a \provider $\beta$, the strategy
$\rhoT:\tupleof{Q',q_0',\delta',o'}$ for $\phiT$ is:
\begin{compactitem}
\item $Q'=Q$ and $q_0'=q_0$,
\item $\delta'(q,\vxs)=\delta(q,\ves)$ where $\ves=\alpha(\vxs)$,
\item $o'(q,\vxs)=\beta(\vxs,c)$ where $c=o(q,\vxs)$.
\end{compactitem}
\end{definition}
We use $C(\alpha,\rhoB,\beta)$ for the combined strategy of $\alpha$,
$\rhoB$ and $\beta$. Now we are ready to state the main theorem.

\begin{theorem}[Correctness of Synthesis Modulo Theories]
  \label{thm:soundness}
  Let $\phiT$ be a realizable specification, $\phiB$ its Boolean
  abstraction, $\alpha$ a \partitioner and $\beta$ a
  \provider.
  Let $\rhoB$ be a winning strategy for $\phiB$, and let
  $\rhoT=C(\alpha,\rhoB,\beta)$ be the combined strategy.
  Then $\rhoT$ is winning for $\phiT$.
\end{theorem}

%\felipe{Check that I didn't mess up the names in the replacement}
\begin{proof} (Sketch).
  Let $\rhoB:\tupleof{Q,q_0,\delta,o}$ and
  $\rhoT:\tupleof{Q,q_0,\delta',o'}$ be the strategies.
  Let $\pi=\seqof{(\xs_0,\ys_0,q_0),(\xs_1,\ys_1,q_1),\ldots}$ be an infinite
  sequence played according to $\rhoT$, that is $\ys_i=o'(q_i,\xs_i)$
  and $q_{i+1}=\delta'(q_i,\xs_i)$.
  Consider the sequence $\seqof{(\Es_0,\Ss_0,q_0)$, $(\Es_1,\Ss_1,q_1),
    \ldots}$ such that $\Es_i=\alpha(\xs_i)$, $\Ss_i=o(q_i,\Es_i)$ and
  $q_{i+1}=\delta(q_i,\Es_i)$.
  Note that this a play of $\phiB$ played according to $\rhoB$ so it
  satisfies $\phiB$.
  In particular, it satisfies $\phiExtra$.
  Moreover, for every time instant $i$, $\ys_i=\beta(\xs_i,\Ss_i)$ by
  construction.
  It follows that, for every $i$, every literal $l_i$ in $\phiT$ and
  the corresponding $s_i$ in $\phiB$ have the same valuation.
  By structural induction, all corresponding sub-formulae of $\phiB$
  and $\phiT$ have the same valuation at every position.
  Therefore, $\pi\models\phiT$. \qed %as desired.  \qed
\end{proof}

\subsection{Standalone Synthesis Modulo Theories}

\subsubsection{Static Provider.}
 
As stated above, a \provider produces, at every step, a model of a
(satisfiable) formula (where some of the elements in the formula are
the inputs received at that specific step).
This can be implemented using an SMT solver at every step.
In this paper we propose an alternative approach where we produce at
static time a \provider via the functional synthesis of a Skolem
function.
The controller then invokes the function produced instead of using
dynamic queries to an SMT solver.
Given an arbitrary relation $R(x,y)$ a Skolem function is a function
$\skh$ that witnesses the validity of $\forall x.\exists y.R(x,y)$ by
guaranteeing that $\forall x..R(x,\skh(x))$ is valid.
Recall that a correct \provider is a function
$\beta:\val(\xs)\times\val(\Ss)\Into\val(\ys)$ which is a witness of
the validity of the following formula:
\[
  \forall \xs. \exists \ys. f_r(\xs) \Into f_c(\xs,\ys)
\]
for a given reaction $r$ and choice $c\in r$.
A Skolem function for $c$ is a function $\skh_c:\val(\xs)\Into\val(\ys)$
such that the following is valid:
\[
  \forall \xs. f_r(\xs) \Into f_c(\xs,h_c(\xs))
\]
For instance, consider a specification where the environment controls
an integer variable $x$ and the system controls an integer variable
$y$ in the specification $\PhiT = \Always (y>x)$.
A Skolem function $\skh(x)=x+1$ serves as a witness (providing values for
$y$) of the validity of $\forall x . \exists y . (\top \Into (y>x))$
and can be used to provide correct integer values for $y$.
For many theories, Skolem functions for $\beta$ can be statically
computed, which means that we can generate statically a \provider for
these theories, and in turn, a full static controller for $\phiT$.
In this paper, we used the \texttt{AEval}
funtional synthesis tool ~\cite{fedyukovich19lazySynthesis}, which
generates witnessing Skolem functions for (possibly many) 
existentially-quantified variables; i.e., \texttt{AEval} will output a function for every
existentially quantified variable: e.g.,
$\forall x \exists y,z. (y>x) \wedge (z>y)$ results in $\skh(x)=x+1$ for $y$
and $\skh(x)=x+2$ for $z$.

\begin{example} \label{ex:functions}
  Consider the strategy Ex.~\ref{ex:provider} where the input
  $e_1:\True$ is mapped to $c_4$ and $e_0:\True$ is mapped to $c_1$.
  The first case requires to synthetize the Skolem function for:
\[
  \forall x. \exists y. (x \geq 2) \Into [(x \geq 2) \wedge (y>1) \wedge (y \leq x)])
\]
whereas the second case requires to handle:
\[
  \forall x. \exists y. (x<2) \Into [(x<2) \wedge (y>1) \wedge (y>x)]),
\]
The corresponding invocations to \texttt{AEval} produce the
following:
\[
  \skh_{(e_1,c_4)} = \begin{cases}
    2 & \text{if $(x \geq 2)$}\\
    0 & \text{otherwise}
  \end{cases}
  \quad\text{and}\quad
 \skh_{(e_0,c_1)} =\begin{cases}
    0 & \text{if $(x \geq 2)$}\\
    x+1 & \text{elif $(1<x)$}\\
    2  & \text{otherwise}\\
  \end{cases}
\]
%
\iffalse
 \[ f_{(e_0,c_4)} = \begin{dcases} \begin{aligned}
 (x \geq 2) & \longrightarrow 2\\
 \text{else} & \longrightarrow 0
 \end{aligned}\end{dcases}
 \quad\text{ and }\quad
 f_{(e_1,c_1)} =
 \begin{dcases} \begin{aligned}
 (x \geq 2) & \longrightarrow 0\\
 \text{else } 1+ & 
 \begin{dcases} \begin{aligned}
 (x > 1) & \longrightarrow x\\
 \text{else}  & \longrightarrow 1
 \end{aligned}\end{dcases},
 \end{aligned}\end{dcases} \]
 \fi
%
where we can see that, when the function $\skh_{(e_1,c_4)}$ is called,
(when $e_1$ holds) only the if branch will hold and will always return
$2$.
Similarly, $\skh_{(e_0,c_1)}$ is called when $e_0$ holds, so only the
else branch will hold since in $\dom(\ThZ)$ it never happens
that $(x>1)$ and $(x<2)$ at the same time. Hence, invocation will
always return $2$.
%\felipe{Wouldn't it always return 1?}
\end{example}

%\newpage

\subsubsection{Predictability.}

Ex.~\ref{ex:provider} shows that there may exist
many different valuations $\vys$ such that $\vys$ matches
the Boolean output trace with values in $\calT$.
The fact that there are many possible outputs that satisfy the same
literals for a given input opens the opportunity to synthesize a
controller for $\phiT$ by adding additional constraints to optimize
certain criteria; e.g., return the greatest value for $\ys$ possible
(we latter study this \textit{adaptivity} in Sec.~\ref{sec:adapt}).

However, since there are many possible outputs that satisfy the same
literals for a given input, using SMT solvers on-the-fly for providing
such outputs does not guarantee that for the same input to the solver,
the same output will be produced.
In practise, different solvers (or even the same solver) can
internally perform different calculations and construct different
models of the same formula even for the same invocation.
This means that the dynamic solver-based approach of
\cite{rodriguez24adaptive} does not guarantee a \provider as a
function (as we present here) but instead it can be non-deterministic:
different invocations with the same input to satisfy the same literals
can produce different outputs.
In other words, the program that implements $\beta$ can be
non-deterministic.
Instead, in our approach with Skolem functions, $\beta$ is a
mathematical function from $\mathbb{B}$ to $\dom(\calT)$ and thus
guarantee \textbf{predictability} in the following sense.

\begin{theorem} [Predictability]
  Let $\phiT$ be a realizable specification, $\phiB$ its Boolean
  abstraction, $\alpha$ a \partitioner and $\beta$ a static \provider.
  Let $\rhoB$ be a winning strategy for $\phiB$, and
  $\rhoT=C(\alpha,\rhoB,\beta)$ be the composed strategy
  Then, given two input traces $\pi_{\xs}$ and $\pi_{\xs}'$ such that 
  $\pi_{\xs}=\pi_{\xs}'$, $\rhoT$ will produce two output traces 
  $\pi_{\ys}$ and $\pi_{\ys}'$ such that $\pi_{\ys}=\pi_{\ys}'$.
\end{theorem}

The theorem follows immediately by $\beta$ being a mathematical
function.
In the next section we extend $\beta$ to provide different outputs for
the same input by explicitly adding arguments to this function.

%\begin{figure}[b!]
%\centering
%  \includegraphics[width=\linewidth]{figures/fig3.pdf}
%  \caption{$\LTLt$ controller using Skolem functions 
%  $\skh_{e_0}$ (for $\skh_{(e_0,c_1)}$) and $\skh_{e_4}$ ($\skh_{(e_1,c_4)}$).}
%  \label{figInterplaySkolem}
%\end{figure}

% \newpage

%%% Local Variables:
%%% TeX-master: "main.tex"
%%% TeX-PDF-mode: t
%%% End:

\section{Adaptive Synthesis Modulo Theories}
\label{sec:adapt}

\subsection{Enhancing Controllers}

%\subsubsection{Time-agnostic adaptivity.}

The static \partitioner presented in the previous section always
generates the same output, for a given choice (valuation of literals)
and input.
However, it is often possible that many different values can be chosen
to satisfy the same choice.
From the point of view of the Boolean controller, any value is
indistinguishable, but from the point of view of the real-world
controller the difference may be significant.
For example, in the theory of linear natural arithmetic
$\Theo = \ThN$, given $x=3$ and the literal $(y>x)$, a Skolem function
$\skh(x)=x+1$ would generate $y=4$, but $y=5$ or $y=6$ are also
admissible.
We call \emph{adaptivity} to the ability of a controller to produce
different values depending on external criteria, while still
guaranteeing the correctness of the controller (in the sense that
values chosen guarantee the specification).
We introduce in this section a \emph{static adaptive} \provider that
exploits this observation.
Recall that the Skolem functions in Sec.~\ref{sec:static} are
synthetised as follows.

\begin{definition}[Basic Provider Formula]
  \label{def:basicF}
  A basic provider formula is a formula of the form
  $\forall \xs. \exists \ys. \psi(\xs,\ys)$, where
  $\psi = f_{r_k(\xs)} \shortrightarrow f_c(\xs,\ys)$ is the
  characteristic formula for reaction $r_k$ and choice $c$.
\end{definition}

We now introduce additional constraints to $\psi$ that---in the case
that the resulting formulae are valid---allow generating functions
that guarantee further properties.
Given a formula $\psi(\xs,\ys)$ and a set of variables $\zs$
(different than $\xs$ and $\ys$) adaptive formulae also enforce an
additional constraint $\psi^+$.
\begin{definition}[Adaptive Provider Formula]
  \label{def:adaptiveF}
  Let $\psi(\xs,\ys)$ be the characteristic formula for a given
  reaction $r_k$ and choice $c$. An \emph{adaptive constraint} is a
  formula $\psi^+(\xs,\zs,\ys)$ whose only free variables are $\xs$,
  $\ys$ and $\zs$.
  An \emph{adaptive provider formula} is of the form
  \[
    \forall \xs, \zs. \exists \ys. [\psi(\xs,\ys) \And \psi^+(\xs,\zs,\ys) ],
  \]
  where $\psi^+$ is an adaptive constraint.
\end{definition}
Note that, in particular, $\psi^+$ can use quantification.
For example, in an arithmetic theory, the additional constraint
$\psi^+:\forall w.\psi(x,w) \Into (|y-z|\leq |w-z|)$ states that all
output alternatives $w$ are farther to $z$ than the $y$ to be
computed.
This formula is constraining the $y$ that must be computed.
The following result guarantees the correctness of using adaptive
provider formulae to craft a provider.
%
%
% Distrae, think where to move it:
% 
% Moreover, we can generate a variety of Skolem functions for each
% reaction and choice pair, and use a different function depending on
% the runtime context within a given mission (e.g., modes of execution).
%

\begin{lemma}
  \label{lem:adaptiveSound}
  Let $\forall \xs,\zs . \exists \ys. (\psi \wedge \psi^+)$ be a
  (valid) adaptive provider formula and let $f$ be a Skolem function
  for it.
  Let $\vxs\in\val(\xs)$ and $\vzs\in\val(\zs)$ be arbitrary
  values.
  Then $\psi(\xs\leftarrow\vxs,\ys\leftarrow f(\vxs,\vzs))$ is
  true.
\end{lemma}

Lemma~\ref{lem:adaptiveSound} shows that synthesizing a Skolem
function for an adaptive formula can be easily transformed into an
Skolem function for the original characteristic formula, so providers
that use adaptive formulae are sound with the original specification.

\begin{example}
  Consider a basic provider formula $\psi:\forall x.\exists y.(y>x)$ in
  $\ThN$ and the Skolem function $\skh(x)=x+1$ generated by
  \texttt{AEval}.
  Consider now the constraint $\psi^+=(y\geq z)\And (y\geq 100)$.
  The adaptive formula
  $\forall x.\forall z.\exists y.(y>x)\And (y\geq z)\And (y\geq 100)$ is valid
  and one possible Skolem function is $\skh(x)=\text{max}(x+1,z,100)$.
  However, if one considers the constraint $\psi_2^+=(y<100)$, the
  resulting provider formula is not valid and there is no Skolem
  function.
  Then, the engineer would have to provide a different constraint or
  use the basic provider formula.
\end{example}  
 
Note that considering different constraints will produce
different Skolem functions without the need of re-synthesizing a
different controller, we only need to switch externally between
functions.

An \emph{adaptive provider description} is a set
$\Gamma=\{\ldots \psi^+_{(r_x,c)}\ldots\}$ of constraint that contains
one constraint per pair $(r_k,c)$---for which $r_k$ is a valid
reaction and $c$ a choice of $r_k$---, such that for every $(r_k,c)$,
the adaptive provider formula
$\forall \xs,\zs. \exists \ys.\psi \And \psi^+_{(r_x,c)}$ is valid.
For example, $\psi^+_{(r_k,c)}=\True$ for every $(r_k,c)$ corresponds
to the basic provider.

%%%%%%%%%
\begin{definition} [Adaptive Provider]
  \label{defAdaptProv}
  Let $\Gamma$ be an adaptive provider description.
  An adaptive \provider is a function
  $\beta_\Gamma:\val(\xs)\times\val(\zs)\times(\Ss)\Into\val(\ys)$ such that
  for every $\vxs\in\val(\xs)$, $\vzs\in\val(\zs)$ and a choice
  $c\in\val(\Ss)$ the following holds:
  \[
    f_{(r_k,c)}(\xs\leftarrow \vxs,\ys\leftarrow \beta(\vxs,\vzs,c))
  \]
\end{definition}

Note that given an adaptive provider description, an adaptive provider
always exists, and is given by any Skolem function for each pair
$(r_k,c)$.

\begin{definition}[Combined Adaptive Strategy]
  \label{def:combinedAdaptive}
  Let $\phiT$ be an \LTLt specification and $\Gamma$ be an adaptive
  provider description.
  Given a \partitioner $\alpha$ for $\phiT$, a controller $\rhoB$ for
  $\phiB$ and an adaptive \provider $\beta_\Gamma$, the strategy
  $\rhoT_\Gamma:\tupleof{Q',q_0',\delta',o'}$ for $\phiT$ is:
\begin{compactitem}
\item $Q'=Q$ and $q_0'=q_0$,
\item $\delta'(q,(\vxs \cup \vzs))=\delta(q,\Es)$ where $\Es=\alpha(\vxs)$,
\item $o'(q,(\vxs \cup \vzs))=\beta_\Gamma(\vxs,\vzs,\Ss)$ where $\Ss=o(q,\Es)$.
\end{compactitem}
\end{definition}

\begin{figure}[t!]
%\begin{minipage}{\textwidth}
\centering
  \includegraphics[width=\linewidth]{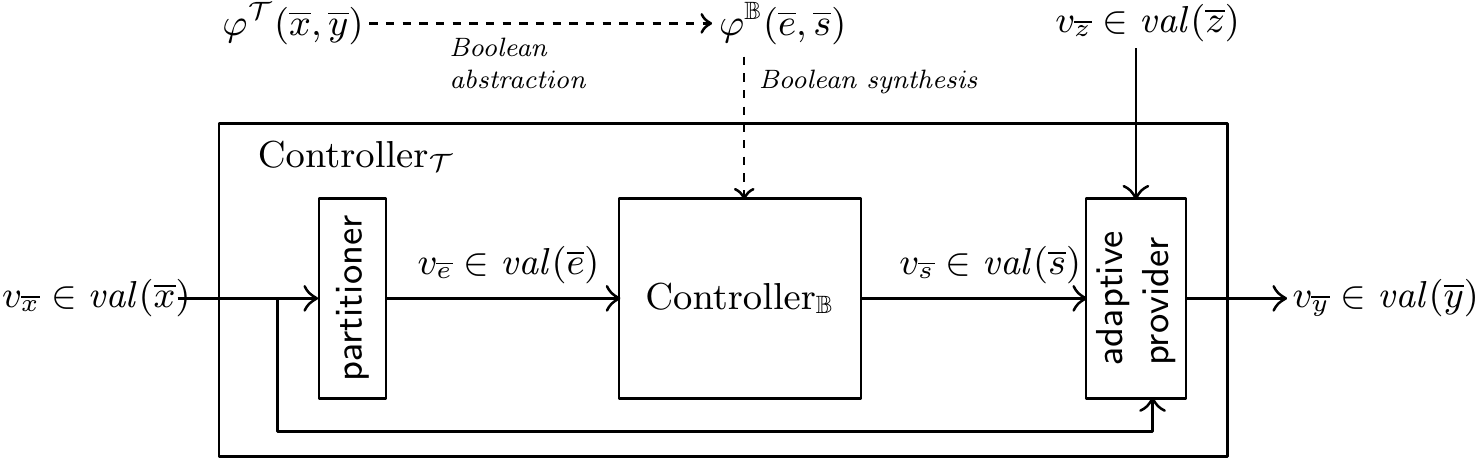}
  \caption{Adaptive architecture, which uses adaptive providers and $\vzs$.}
  \label{figAdaptive}
\end{figure}

Note that in the semantics of $\phiT$ now the environment chooses the
values $\vzs$ of the variables $\zs$ that appear in the constraints in
$\Gamma$. 
Also, note that the overall architecture of adaptive controllers
(see Fig.~\ref{figAdaptive})
is similar to the one presented in Sec.~\ref{sec:static}, 
but using adaptive providers and extra input $\vzs$.

The following holds analogously to Thm.~\ref{thm:soundness} in
Sec.~\ref{sec:static}.

\begin{theorem}[Correctness of Adaptive Synthesis]
  \label{thm:soundnessAdaptive}
  Let $\phiT$ be a realizable specification, $\phiB$ its Boolean
  abstraction, $\alpha$ a \partitioner, $\Gamma$ an adaptive provider
  description and $\beta_\Gamma$ an adaptive \provider.
  Let $\rhoB$ be a winning strategy for $\phiB$, and $\rhoT_\Gamma$
  the strategy obtained as the composition of $\alpha$, $\rhoB$ and
  $\beta_\Gamma$ described in Def.~\ref{def:combinedAdaptive}.
  Then $\rhoT_\Gamma$ is winning for $\phiT$.
\end{theorem}

\begin{proof}[Sketch]
  The proof proceeds by showing that any play played according to
  $\rhoT_\Gamma$ satisfies, at all steps, the same literals as
  $\rhoB$, independently of the values of $\zs$.
  It is crucial that, after each step, the controller state that
  $\rhoT_\Gamma$ and $\rhoB$ leave is the same, which holds because
  their $\delta'$ is indistinguishable. \qed
\end{proof}

Skolem functions are computed from basic provider formulae that have a
shape $\forall^*\exists^*. \psi$.
This shape is preserved in adaptive provider formulae in which the
constraint $\psi^+$ is quantifier-free.
However, as the last example illustrated, the constraint $\psi^+$ may
include quantifiers, which does not preserve the shape typically
amenable for Skolemization.

For instance, to compute the smallest $y \in \dom(\ThZ)$ in
$\forall x. \exists y. (y>x)$, one can use the adaptive provider
formula
$\forall x. \exists y. [(y>x) \wedge \forall z. (z>x) \Into
(z\geq{}y)] $.
We overcome this issue by performing quantifier elimination (QE) for
the innermost quantifier and recover the $\forall^*\exists^*$
shape.
In consequence, our resulting method for adaptive provider generation
works on any theory $\calT$ that:
\begin{compactenum}[(1)]
\item is decidable for the $\exists^*\forall^*$ fragment (for the
  Boolean abstraction);
\item permits a Skolem function synthesis procedure (for valid
  $\forall^*\exists^*$ formulae), for producing static providers; and
\item accepts QE (which preserves formula equivalence) for the
  flexibility in defining quentified constraints $\psi^+$.
\end{compactenum}

\begin{example} \label{ex:functionsAdaptive}
  Consider again the strategy of Ex.~\ref{ex:provider} and
  the first Skolem function to synthetise at Ex.~\ref{ex:functionsAdaptive}:
$
  \forall x. \exists y. (x \geq 2) \Into \psi,
$
where $\psi=[(x \geq 2) \wedge (y>1) \wedge (y \leq x)]$.
  Then, we add the adaptivity criteria that we want our strategy to return the greatest value \textbf{possible},
  so the function to synthetise is as follows:
\[
  \forall x. \exists y. (x \geq 2) \Into (\psi \wedge 
  \forall z. [(x \geq 2) \wedge (z>1) \wedge (z \leq x) \Into (z \leq y)])
\]
We show below the results of \texttt{AEval} invocations with the original (left)
and the adaptive (right) versions:
\[
  \skh_{(e_1,c_4)}(x) = \begin{cases}
    2 & \text{if $(x \geq 2)$}\\
    0 & \text{otherwise}
  \end{cases}
  \quad\text{and}\quad
  \skh_{(e_1,c_4)}^+(x) = \begin{cases}
    x & \text{if $(x \geq 2)$}\\
    0 & \text{otherwise}
  \end{cases}
\]
where we can see that, $\skh_{(e_1,c_4)}$ is a more static function in the sense
that it will always return $2$, whereas $\skh_{(e_1,c_4)}^+$ depends on the value
of $x$ in order to return exactly $x$ (which is the greatest value possible).
\end{example}
\section{Empirical Evaluation}
\label{sec:empirical}

We now report on empirical evaluation to asses the performance of our
approach.
We used Python $3.8.8$ for the implementation of the architecture and
Z3 $4.12.2$ for the SMT queries.
We use Strix~\cite{meyerETAL2018strixSynthesisStrikesBack} as the
synthesis engine and \texttt{aigsim.c} to execute the synthetised
controller.
For functional synthesis we used the \texttt{AEval}
solver~\cite{fedyukovich19lazySynthesis}\footnote{Publicly available
  at: \url{https://github.com/grigoryfedyukovich/aeval}} that
leverages Z3.
Currently, \texttt{AEval} expects formulae in linear arithmetic with
the $\forall^*\exists^*. \varphi$ shape, which is suitable for the
static provider we want to synthesise.
We translate the Skolem functions into C\texttt{++} and used
\texttt{g++} $14.0.0$ as a compiler.
We ran all experiments on a MacBook Air $12.4$ with the M1 processor
and $16$ GB. of memory.
%
%All experiments could be replicated for the corresponding \LTLft
%benchmarks by choosing an appropriate synthesis engine.

\subsubsection{Wrap-up experiment.}
We first report our results on $\calT$-controller for
Ex.~\ref{exRunning}.
Following the idea of Ex.~\ref{ex:provider}, we execute 
the input trace  
$\pi = \seqof{\record{x|x \geq 2}, \record{x|x \geq 2}, 
\record{x|x < 2}, \record{x|x < 2}, \record{x|x \geq 2}}$
$100000$ times on (1) a \textbf{dynamic} \provider
following~\cite{rodriguez24adaptive} and (2) our \textbf{static}
\provider approach.
Throughout both experiments, the average time for the \partitioner was
$28$ ms\footnote{Note that the time is dominated by the \partitioner,
  shared in both cases, which consists on searching among a finite
  collection of formulae (valid reactions) to find the correct
  partition and can be easily optimized.} and the average time for the
Boolean controller execution was $2.47$ $\mu$s.
However, the average time for the \textbf{dynamic} \provider was $169$
$\mu$s, whereas the \textbf{static} \provider was about $50$ times
faster: $2.9$ $\mu$s.
%
%The total running time of the first experiment was $24122$s and
%the second one is $24122$s. 
%
% In order to better illustrate the evolution of the experiments, since
% $\pi$ is too long (thus its plot is too dense),
We show in Fig.~\ref{fig:wrapProvTimes} the time needed (in $\mu$s) of
the \textbf{dynamic} \provider and the \textbf{static} \provider in
the first $50$ events.
We can see that (1) the times required in the \textbf{dynamic}
approach are more unstable and that (2) the \textbf{dynamic} approach
is two orders of magnitude faster.
Fig.~\ref{fig:provDynamic} and Fig.~\ref{fig:provStatic} zoom over
Fig.~\ref{fig:wrapProvTimes}.

\begin{figure}[h!]
\minipage{0.33\textwidth}
  \includegraphics[width=1.1\linewidth]{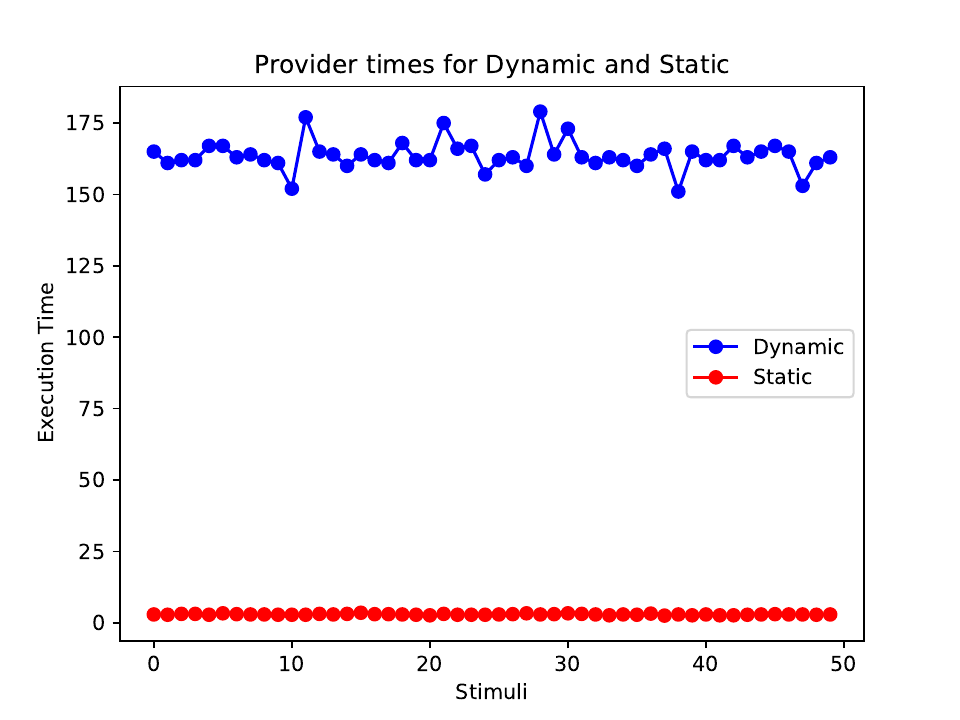}
  \caption{Comparison.}
  \label{fig:wrapProvTimes}
\endminipage\hspace{-0.2em}
\minipage{0.33\textwidth}
  \includegraphics[width=1.1\linewidth]{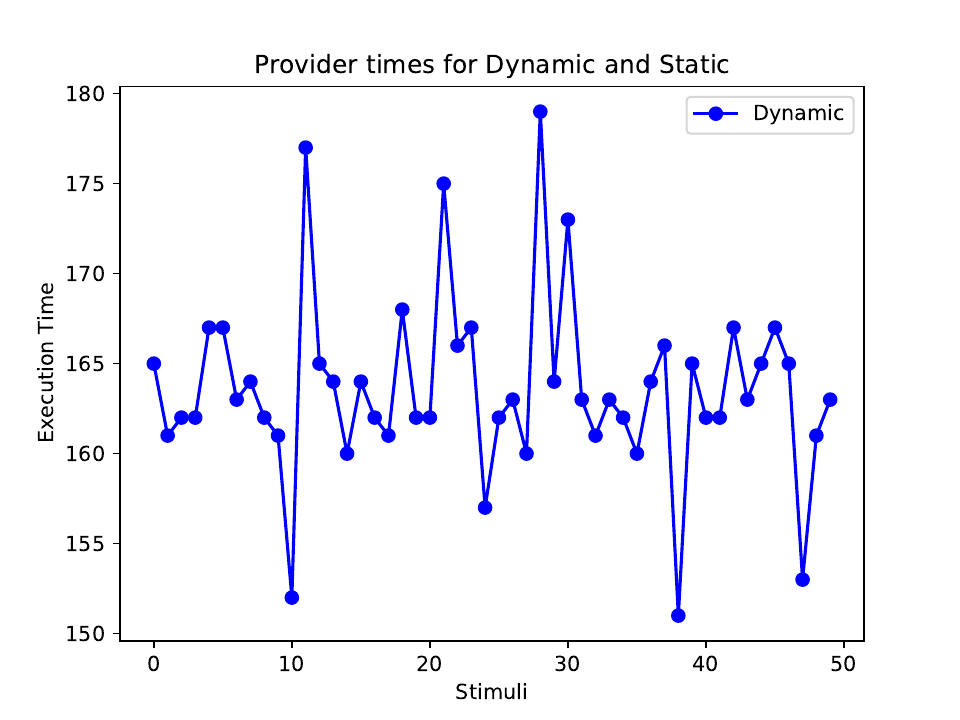}
  \caption{Zoom in Dynamic.}
  \label{fig:provStatic}
\endminipage\hspace{-0.2em}
\minipage{0.33\textwidth}%
  \includegraphics[width=1.1\linewidth]{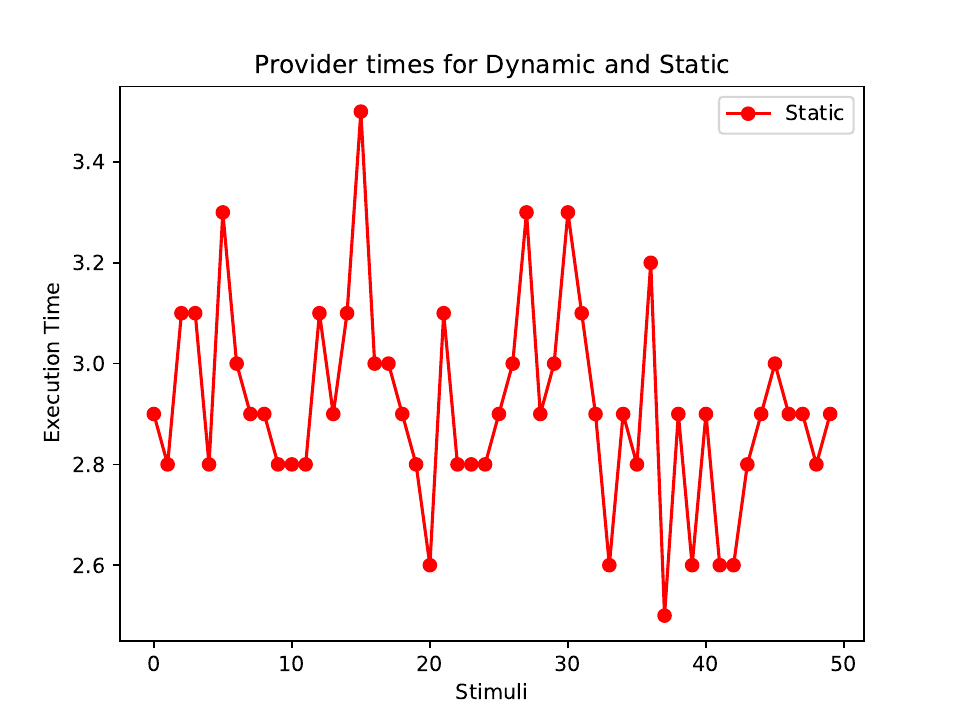}
  \caption{Zoom in Static.}
  \label{fig:provDynamic}
\endminipage
\end{figure}
It is an important detail that, since $\pi$ only provides ranges of
inputs (e.g., a general $x|(x<2)$ instead of a concrete $x:1$), the
input values may be different in both experiments.
Therefore, we executed again the experiments over a same fixed input
trace
$\pi' = \seqof{\record{x:4}, \record{x:4}, \record{x:1}, \record{x:0},
  \record{x:2},\ldots}$ in order to do a sanity check.
Fig.~\ref{fig:provDynamicFixed} shows that the execution with $\pi'$
follows the same tendency as with $\pi$.
Even though the input numbers are repeated, we still encounter
differences in solving times both in Fig.~\ref{fig:provStaticFixed}
and Fig.~\ref{fig:wrapProvTimesFixed}, which suggests that, for such
an amount of constraints to solve, the time to solve is not dominated
by the input, but rather by implementation and memory details.

\begin{figure}[t!]
\minipage{0.33\textwidth}
  \includegraphics[width=1.1\linewidth]{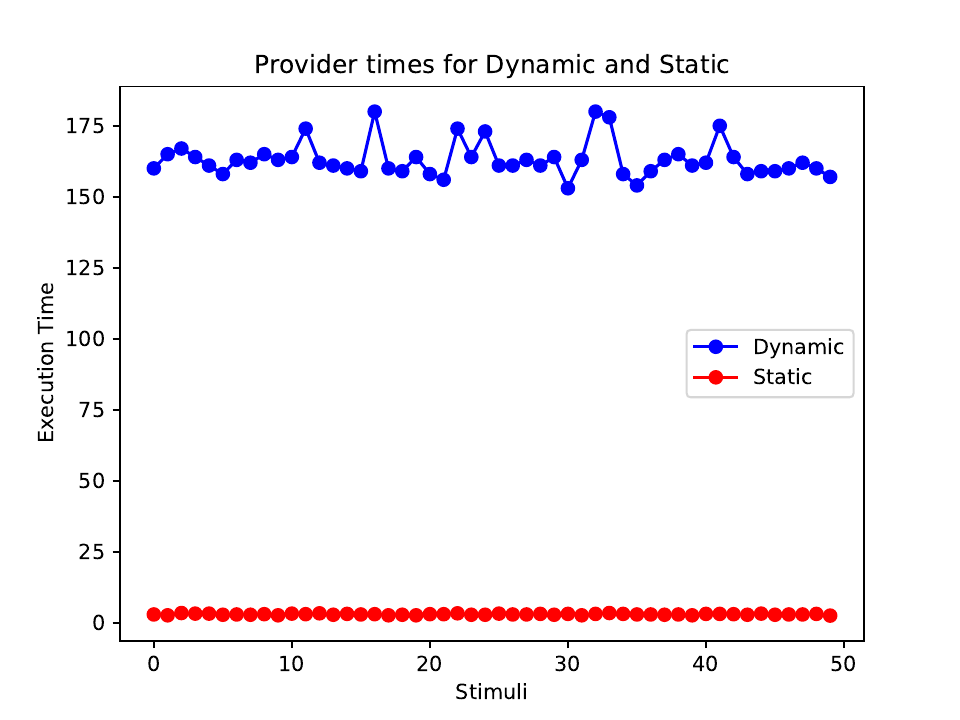}
  \caption{Comparison.}
  \label{fig:provDynamicFixed}
\endminipage\hspace{-0.2em}
\minipage{0.33\textwidth}
  \includegraphics[width=1.1\linewidth]{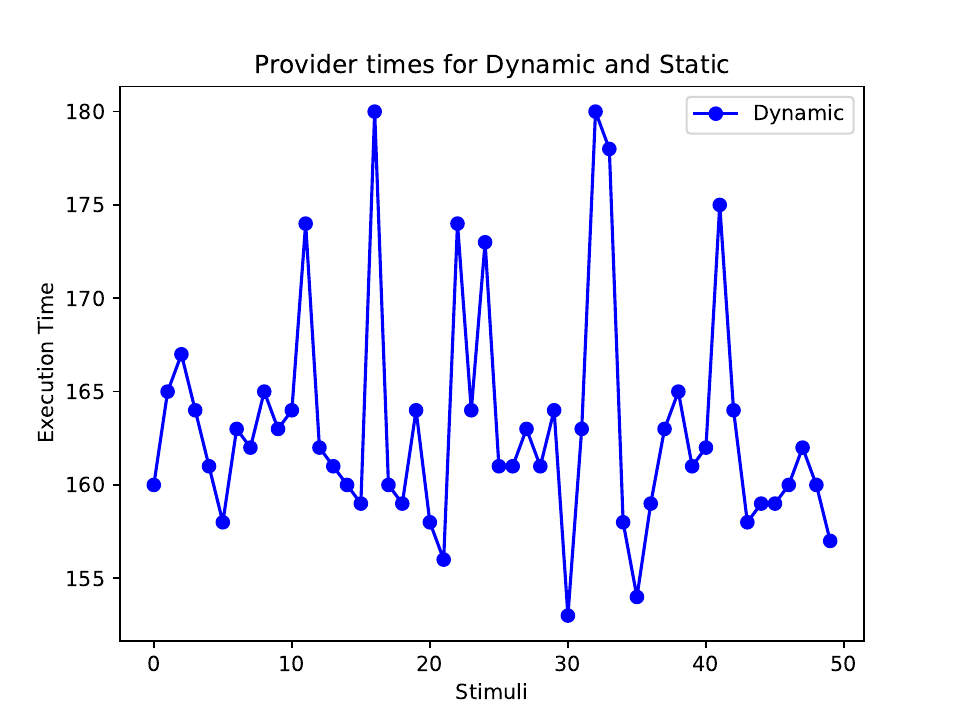}
  \caption{Zoom in Dynamic.}
  \label{fig:provStaticFixed}
\endminipage\hspace{-0.2em}
\minipage{0.33\textwidth}%
  \includegraphics[width=1.1\linewidth]{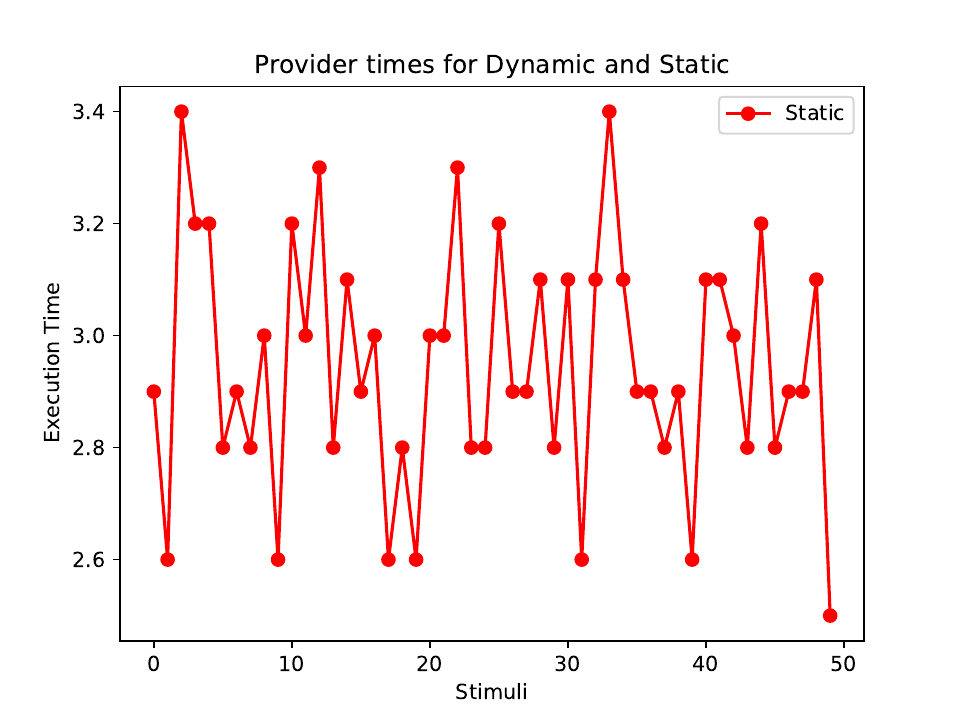}
  \caption{Zoom in Static.}
  \label{fig:wrapProvTimesFixed}
\endminipage
\end{figure}

We also checked the predictability of both approaches using $\pi'$:
i.e., how much does the output differ given the same input and
position in the game.
Recall from Ex.~\ref{ex:provider} and $\pi'$ that at timestep $t:0$
the possible outcomes are $\vys \in \{2,3,4\}$, at $t:1$ again
$\vys \in \{2,3,4\}$, at $t:2$ $\vys \in \{2,3,4,...\}$, at $t:3$
again $\vys \in \{2,3,4,...\}$ and at $t:4$ $\vys \in \{2\}$.
Let us denote with $k$ the repetition of this pattern.
At time $t:0+k$ and $t:1+k$ there are three valid outputs, at $t:2+k$
and $t:3+k$ there are infinitely many valid outputs and at $t:4+k$
there is a single valid output.
Also, note that $\vys$ should be a valid output for every input in $\pi'$
(as captured by both Skolem functions in
Ex.~\ref{ex:functions}).
In Fig.~\ref{fig:predictability}, we show results for $500$ timesteps
($100 \times |\pi'|$).
We can see that the dynamic \provider is less predictable (outputs),
whereas the static \provider always produces the value $2$.
Note that output values in the dynamic \provider are always different
in every experiment, but the general shape remains similar.
Also, note that one might consider that predictability in the dynamic
approach is also remarkably stable, since only for $20$ times out of
the total of $400$ the \provider produces a value different than $2$
($17$ times value $3$, twice the value $4$ and once value $5$).

However, this stability difference increases when adaptivity is
considered.
Concretely, we will use \textit{return the greatest} (illustrated in
Ex.~\ref{ex:functionsAdaptive}) for $t:0+k$, $t:1+k$ and $t:0+4$, and
\textit{return the smallest} for $t:2+k$ and $t:3+k$, which means that
the \textit{ideal} output trace with respect to these criteria is the
pattern
$\seqof{\record{y:4},\record{y:4},\record{y:2},\record{y:2},\record{y:2}}$.
As expected, the static \provider always returns the
ideal pattern, whereas the pattern in the dynamic case is more
unstable.
For example, in Fig.~\ref{fig:predictabilityAdaptive}, we show results
for $50$ stimuli in $\pi'$, where three times the output was not
within the ideal pattern.
We acknowledge that, in the dynamic approach, as the timeout
restrictions for the underlying SMT solver gets more strict (e.g., in
fast embedded contexts), less adaptive constraints will be solved and
thus the output will tend to diverge more from the ideal pattern.

\begin{figure}[!htb]
\minipage{0.33\textwidth}
  \includegraphics[width=1.1\linewidth]{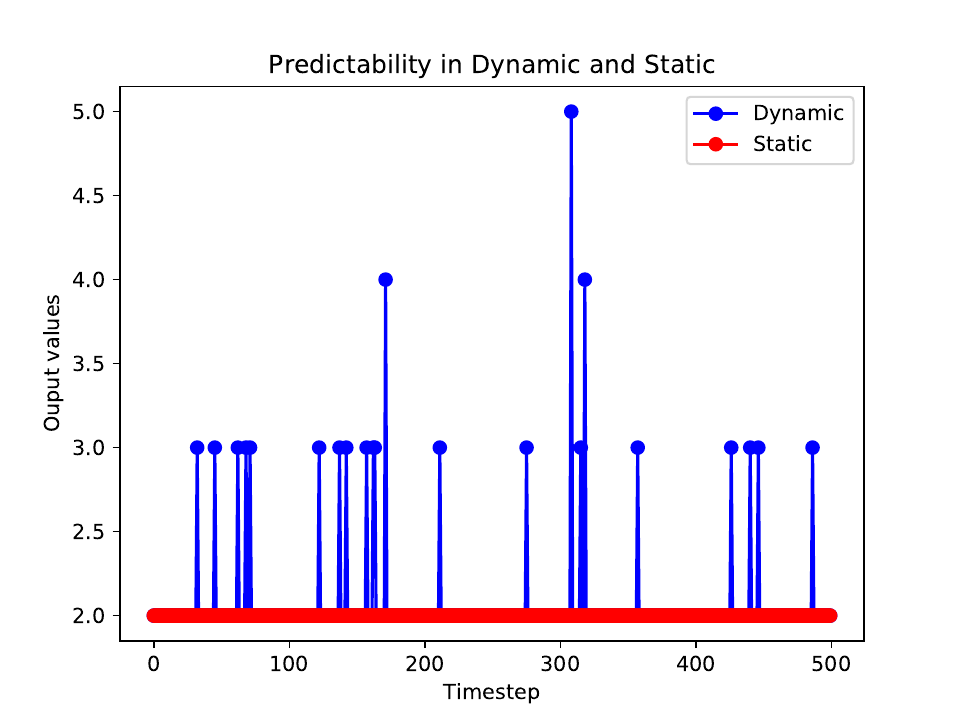}
  \caption{Comparison.}
  \label{fig:predictability}
\endminipage\hspace{-0.1em}
\minipage{0.33\textwidth}
  \includegraphics[width=1.1\linewidth]{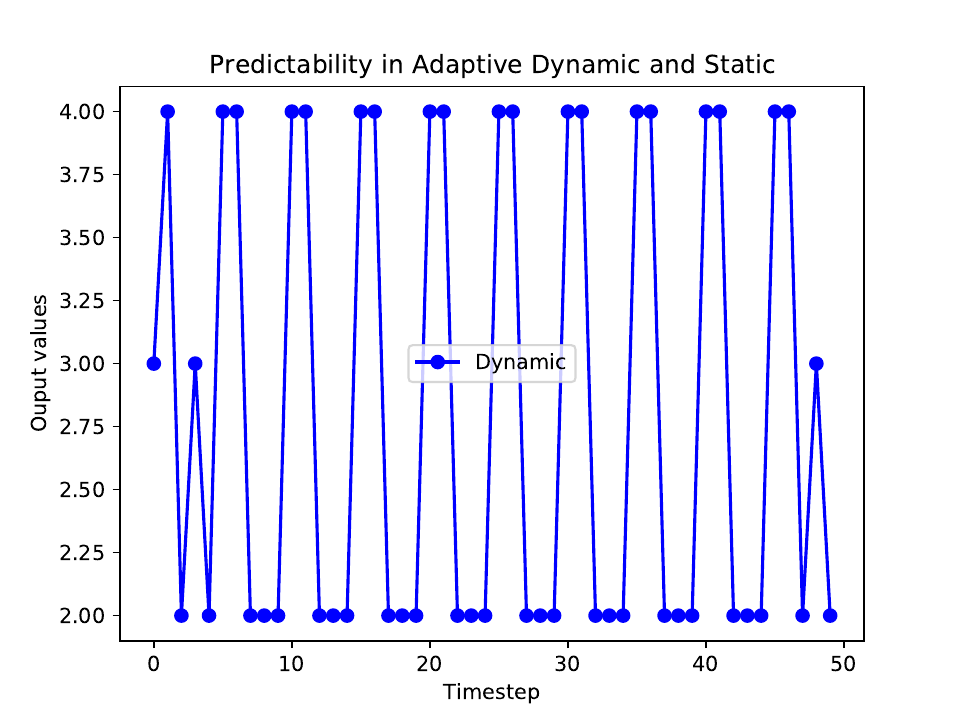}
  \caption{Dyn. Adaptive.}
  \label{fig:predictabilityAdaptive}
\endminipage\hspace{-0.1em}
\minipage{0.33\textwidth}%
  \includegraphics[width=1.1\linewidth]{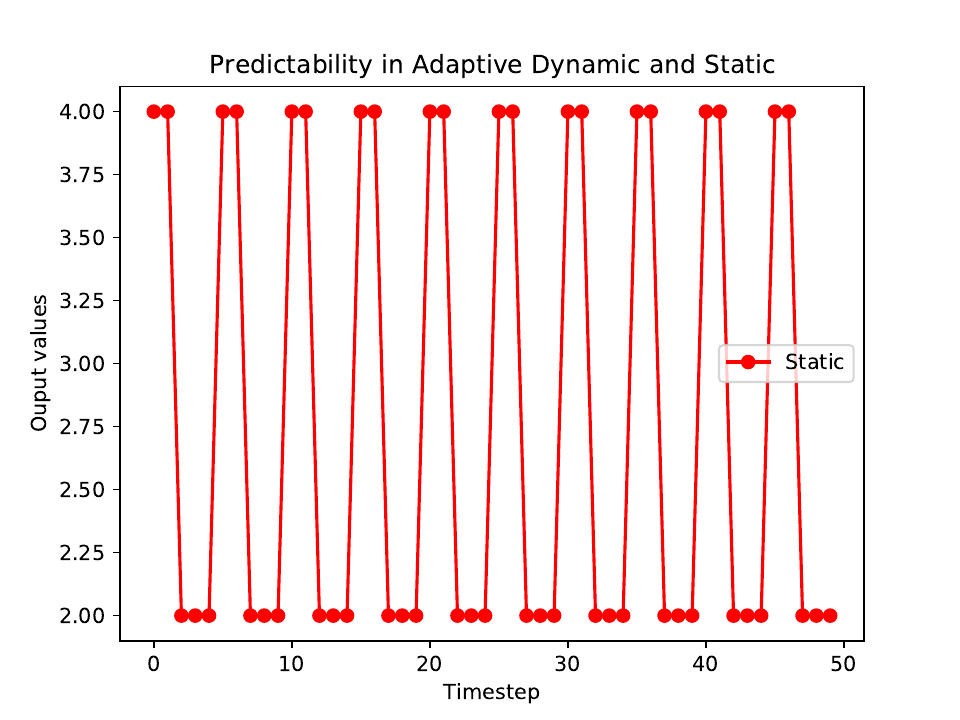}
  \caption{Static Adaptive.}
  \label{fig:predictabilityAdaptiveTimeout}
\endminipage
\end{figure}

% \iffalse
% %TO PUT 2 IMAGES
% \begin{figure*}[t!]
% %\hspace{-2.5em}
% \centering
% \begin{tabular}{c@{\hspace{1em}}c}
%   \includegraphics[scale=0.25]{predictability_dynamicStatic_times.pdf} &
%   \includegraphics[scale=0.25]{predictability_dynamicStatic_times.pdf} \\
%   (Not-adaptive) Comparison. & (Adaptive) Dynamic zoom. 
% \end{tabular}
% \caption{Architectures for shields modulo theory.}
% \label{fig:predictability}
% \end{figure*}
% \fi

\subsubsection{Benchmark results.}

% \iffalse ABOUT CLUSTERS AND PREPRO
% Recall from \cite{rodriguez23boolean} that if the literals in $\varphi$ are split into clusters that do not share
% variables, the Boolean abstraction process can handle each cluster
% independently, and one can use \provider and \partitioner
% independently.
% %
% Thus, clusters
% of literals that do not share variables were Booleanized independently
% and composed afterwards.
% %
% For each example we show the number of variables (\textit{vr.}) and
% literals (\textit{lt.}) per cluster.
% %
% The first column shows the name of the benchmarks (\textit{nm.}).
% %
% The Boolean controller synthesis columns (\textit{BA}) report the time
% needed for the Boolean abstraction for each cluster (\textit{Tme.}),
% together with the number of SMT queries necessary (\textit{Quer.}),
% the number of input decisions available for the environment
% ($\textit{Dc.}$) and the shared Boolean controller synthesis time
% ($\Boolbb$).
%\fi

In order to perform a wide empirical evaluation, we used benchmarks
from \cite{rodriguez24adaptive} to validate the hypothesis whether
that our static approach is faster and more predictable.
Tab.~\ref{tabBenchmark} shows the main group of experiments.
\textit{Sz.} refers to the size of the $\phiT$ specification both in
variables (\textit{vr.}) and literals (\textit{lt.}).  \textit{Prep.}
refers to the pre-processing time (i.e., the time needed to compute
the Boolean abstraction and synthetize a Boolean controller).
This and the computation of the \partitioner are performed at compile
time, while in our experiments \provider is computed dynamically for
each new $(\ves,\Vss)$ pair discovered.
Note that the time necessary to compute Skolem functions was
negligible (around $2$ seconds) as expected, since the number of
constraints we used does not stress \texttt{AEval} \footnote{Indeed,
  an eventual input with a larger amount of constraints (i.e.,
  literals) may or may not yield a bottleneck for the underlying
  Boolean abstraction procedure, but not for the approach we present
  in this paper.}.
The following two groups of columns show results of the execution of
$1000$ ($1$K) and $10000$ ($10$K) timesteps of input-output
simulations.
For each group we do not show the average time that the \partitioner
takes to respond with a discrete $\ves$ from $\vxs$ and the average
time that the Boolean controller takes to respond with a Boolean
$c_i$, since it is the same for the \textbf{dynamic} and the
\textbf{static} approaches.
Instead, we report the time that the \provider takes to respond with a
valuation associated to $c_i$ in \textbf{dynamic} (\textit{Dyn.}) and
\textbf{static} (\textit{St.}).
Note that \textit{Tr.} is also the benchmark that takes the most time
on average for $\textit{Pr.}$, since its functions contain more
operations, but it is still efficient enough for the targeted
applications.
Overall, we can see that the static approach is $50$ to $60$ times
faster.

In addition, we used adaptivity with different criteria.
We selected Skolem functions of each benchmark and created adaptive
versions that return: (1) minimum and maximum valuation possible
(\textit{m/m}) knowing there was at least one of such bounds to stop
the search and (2) valuation closest to a $p$ point (\textit{pc.})  that
is randomly generated.
Then, groups ($10$K (\textit{m/m})) and ($10$K (\textit{pc.})) measure
again average times of the dynamic \provider and adaptive \provider.
Note that it is not clear how adaptivity impacts \provider time, since
sometimes the amount of adaptive constraints dominates this time,
whereas in other cases the complexity of the criteria (e.g.,
underlying quantified structure) seems to be more relevant.

\TableBenchmark 
Moreover, we show approximated predictability (\textit{Pre.}) results
of the dynamic approach as a percentage that measures how many of the
outputs diverged from the ideal pattern (e.g., we previously mentioned
that in Fig.~\ref{fig:predictability} this number was $\sim 5\%$: $20$
divergences out of $400$ stimuli).
In Tab.~\ref{tabBenchmark} we can see that predictability is between
$2$ and $15$ percent (average about $5\%$), whereas in the static
varsion it is $0\%$ (not shown in the table).
More importantly, it seems that adding adaptivity tends to result in
less predictable outputs in the dynamic \provider.
These results support our hypothesis that our static controller is
more efficient and predictable.

We also tested whether $\calT$ affected the performance of the
controller, via use cases \textit{Syn. (2,3)} to \textit{Syn. (2,6)}
of \cite{rodriguez23boolean} interpreted over linear integer
arithmetic and linear real arithmetic (the theories accepted by
\texttt{AEval}).
%
%RESULTS FROM: https://drive.google.com/drive/folders/17Js4TL9F1Ci4HH5qtWlH6FkDfkypbslW
\begin{figure}[b!]
  \centering
%\begin{center}
\begin{tabular}{|c|c c c|c c c|} 
 \hline
 \multirow{2}{*}{Lits} & \multicolumn{3}{|c|}{Linear I. Arithmetic} & \multicolumn{3}{c|}{Linear R. Arithmetic} \\
 & 1K (bs.) & 1K (m/m) & 1K (pc.) & 1K (bs.) & 1K (m/m) & 1K (pc.) \\
 %Literals & Performed queries & Out of & Needed queries \\ 
 %(2 to 12) & (outer+inner) & (exponential) & ($\simeq\%$) \\ [0.5ex] 
 \hline\hline
 $3$ & $3.41$ & $2.32$ & $2.00$ & $3.45$ & $2.14$ & $2.02$\\ 
 \hline
 $4$ & $3.50$ & $2.05$ & $2.07$ & $198$ & $3.60$ & $2.09$ \\
 \hline
  $5$ & $3.63$ & $2.64$ & $2.63$ & $3.83$ & $2.65$ & $2.67$ \\
 \hline
  $6$ & $3.78$ & $3.30$ & $3.38$ & $3.85$ & $3.39$ & $3.45$ \\
 \hline

\iffalse
SMT results:
 \hline\hline
 $3$ & $188$ & $116$ & $110$ & $190$ & $118$ & $110$\\ 
 \hline
 $4$ & $193$ & $113$ & $114$ & $198$ & $112$ & $115$ \\
 \hline
  $5$ & $200$ & $145$ & $145$ & $211$ & $146$ & $147$ \\
 \hline
  $6$ & $208$ & $182$ & $186$ & $212$ & $186$ & $190$ \\
 \hline
 \fi

\end{tabular}
\caption{Comparison of $\ThZ$ and $\ThR$ for \textit{Syn (2,3)} to \textit{Syn (2,6).}}
% \caption{A preliminar comparison of double-SAT with \textit{Synt} 3 to
%   6 interpreted in $\mathcal{T}_{\mathbb{Z}}$ and interpreted in
%   $\mathcal{T}_{\mathbb{R}}$. We observe results are similar in
%   $\mathcal{T}_{\mathbb{Z}}$ and $\mathcal{T}_{\mathbb{R}}$.}
\label{tab:theoryComparison}
%\end{center}
%\vspace{-1cm} %too close with 1cm and it does not allow to use 0.7
\end{figure}
We show the results in Fig.~\ref{tab:theoryComparison}, where we compare
$1000$ simulations in the basic case (\textit{bs.}), and the adaptive 
minimal/maximal (\textit{m/m.}) and randomly generated point (\textit{pc.}) cases.
Note that
run-time difference is negligible and so was in the abstraction phase.
Also, note that Skolem function synthesis was slightly harder in integers.
%
%Also, note that for smallest/greatest, we used an $\epsilon$ distance. %, as mentioned in Sec.~\ref{sec:static}.

%%% Local Variables:
%%% TeX-master: "main.tex"
%%% TeX-PDF-mode: t
%%% End:

%\newpage
\section{Related Work and Conclusions}
\label{sec:conclusion}

\subsubsection{Related Work.}

LTL modulo theories has been previously studied
(e.g., \cite{gianola2022ltl,faran2022ltl}), but allowing
temporal operators within predicates, again leading to undecidability.
Also, infinite-state synthesis has been recently studied
at~\cite{cheng2013numerical,azadeh2017strategy,gacek2015towards,samuel23symbolic,azzopardi2023ltl,heim24solving} 
but with similar restrictions. 
At \cite{katis2016synthesis,katis2018validity} authors perform
reactive synthesis based on a fixpoint of $\forall^*\exists^*$
formulae (for which they use \texttt{AEval}), but expressivity is
limited to safety and does not guarantee termination.
The work in \cite{walker2014predicate} also relies on abstraction but
needs guidance and again expressivity is limited.
Reactive synthesis of Temporal Stream Logic (TSL) modulo
theories~\cite{finkbeinerETAL2021temporalStreamLogicModuloTheories} is
studied in~\cite{wonhyuk2022synthesis,maderbacher2021reactive}, which
extends LTL with complex data that can be related accross time.
Again, note that general synthesis is undecidable by relating values across time.
Moreover, TSL is already undecidable for safety, the theory of
equality and Presburger arithmetic.
Thus, all the specifications considered for empirical evaluation in
Sec.~\ref{sec:empirical} are not within the considered decidable fragments.

All approaches above adapt one specific technique and
implement it in a monolithic way, 
whereas \cite{rodriguez23boolean,rodriguez24realizability} generates \LTL specification that existing tools can
process with any of their internal algorithms (bounded
synthesis, for example) so we will automatically benefit from further
optimizations in these techniques.
Moreover, Boolean abstraction preserves the temporal fragments like safety
and GR(1) so specialized solvers can be used.
Throughout the paper, we have already extensively compared the work \cite{rodriguez24adaptive} with ours
and we showed that our approach uses Skolem functions instead of SMT queries on-the-fly, 
which makes it faster, more predictable and a pure controller that can be used in embedded contexts.
It is worth noting that \cite{rodriguez24adaptive} and our approach can be understood 
as computing \textit{minterms} to produce Symbolic automata and transducers
\cite{dAntoni14minimization,dAntoni2017power}
from reactive specifications 
(and using antichain-based optimization, as suggested by \cite{veanes2023symbolic}).
Also, note that any advance in abstraction method (e.g., \cite{azzopardi23ltl}) 
has an immediate positive impact in our work.

Last, \cite{qinheping17automatic} presents a similar idea to our
Skolem function synthesis: instead of solving a quantified formula every time one 
wants to compute an output, they synthesize a term that computes the output from the input.
However, the paper is framed in the program synthesis problem and uses syntax-guided synthesis
\cite{alur13syntax}, whereas previous reactive synthesis papers
have suggested functional synthesis as a recommended software engineering practise
(e.g., \cite{samuel21genSys}).

\subsubsection{Conclusion.}
The main contribution of this paper is the synthesis procedure for
\LTLt, using internally a Boolean controller and static Skolem
function synthesis, which is more performant and predictable than
previous approaches.
Our method also allows producing \emph{adaptive responses} that optimize
the behaviour of the controller with respect to different criteria.
%even
%taking into consideration past values of inputs and outputs, providing
%correct and smooth outputs via combination with external programs that
%provide information that cannot be incorporated easily within \LTLt
%controller synthesis.
%
%A notable instance is neurosymbolic reactive synthesis where an ML
%component can be used externally to provide values that our controller
%aims to approximate.
%
We showed empirically that our approach is fast for many targeted
applications and analyzed the cost and predictability of our Skolem functions
component compared to \cite{rodriguez24adaptive}.
As far as we know, this is the first decidable full reactive synthesis approach (with
or with adaptivity) for \LTLt specifications.

Future work includes first to use winning regions instead of concrete
controllers to allow even more choices for the Skolem functions, and
to develop a further adaptivity theory.
%
%Also, we want to compare the use of Skolem functions with calling SMT
%solvers on the fly, which allows testing semi-decidable theories.
%
Another direction is studying adaptivity over the environment inputs
and combining this approach with monitoring.
Also, we plan to study how to extend \LTLt with transfer
of data accross time preserving decidability, since
recent results~\cite{geatti23decidable} suggest that the expressivity
can be extended with limited transfers in semantic fragments of $\LTLt$.
Moreover, explaining our synthesis approach within more general frameworks
like (e.g., \cite{geatti24general}) is immediate work to do.
Finally, we want to study how to use our approach to construct
more predictable and performant shields \cite{alshiekhETAL2017safeReinforcementLearningShielding,bloem15shield}
(concretely, shields modulo theories \cite{rodriguez24shield,corsi24verification})
to enforce safety in critical systems.

%%% Local Variables:
%%% TeX-master: "main.tex"
%%% TeX-PDF-mode: t
%%% End:

%\vfill
%\pagebreak

\bibliographystyle{plain}
\bibliography{references}

\begin{thebibliography}{10}

\bibitem{alshiekhETAL2017safeReinforcementLearningShielding}
Mohammed Alshiekh, Roderick Bloem, R{\"{u}}diger Ehlers, Bettina
  K{\"{o}}nighofer, Scott Niekum, and Ufuk Topcu.
\newblock Safe reinforcement learning via shielding.
\newblock In Sheila~A. McIlraith and Kilian~Q. Weinberger, editors, {\em Proc.
  of the Thirty-Second {AAAI} Conference on Artificial Intelligence,
  (AAAI-18)}, pages 2669--2678. {AAAI} Press, 2018.

\bibitem{alur13syntax}
Rajeev Alur, Rastislav Bod{\'{\i}}k, Garvit Juniwal, Milo M.~K. Martin, Mukund
  Raghothaman, Sanjit~A. Seshia, Rishabh Singh, Armando Solar{-}Lezama, Emina
  Torlak, and Abhishek Udupa.
\newblock Syntax-guided synthesis.
\newblock In {\em Proc. of the 13th Int'l Conf. on Formal Methods in
  Computer-Aided Design ({FMCAD} 2013)}, pages 1--8. {IEEE}, 2013.

\bibitem{azzopardi2023ltl}
Shaun Azzopardi, Nir Piterman, Gerardo Schneider, and Luca di~Stefano.
\newblock Ltl synthesis on infinite-state arenas defined by programs, 2023.

\bibitem{azzopardi23ltl}
Shaun Azzopardi, Nir Piterman, Gerardo Schneider, and Luca~Di Stefano.
\newblock {LTL} synthesis on infinite-state arenas defined by programs.
\newblock {\em CoRR}, abs/2307.09776, 2023.

\bibitem{bloem12synthesis}
Roderick Bloem, Barbara Jobstmann, Nir Piterman, Amir Pnueli, and Yaniv Sa'ar.
\newblock Synthesis of reactive(1) designs.
\newblock {\em J. Comput. Syst. Sci.}, 78(3):911--938, 2012.

\bibitem{bloem15shield}
Roderick Bloem, Bettina K{\"{o}}nighofer, Robert K{\"{o}}nighofer, and Chao
  Wang.
\newblock Shield synthesis: - runtime enforcement for reactive systems.
\newblock In {\em Proc. of the 21st International Conference in Tools and
  Algorithms for the Construction and Analysis of Systems ({TACAS} 2015)},
  volume 9035 of {\em LNCS}, pages 533--548. Springer, 2015.

\bibitem{bradley07calculus}
Aaron~R. Bradley and Zohar Manna.
\newblock {\em The Calculus of Computation}.
\newblock Springer-Verlag, 2007.

\bibitem{cheng2013numerical}
Chih{-}Hong Cheng and Edward~A. Lee.
\newblock Numerical {LTL} synthesis for cyber-physical systems.
\newblock {\em CoRR}, abs/1307.3722, 2013.

\bibitem{wonhyuk2022synthesis}
Wonhyuk Choi, Bernd Finkbeiner, Ruzica Piskac, and Mark Santolucito.
\newblock Can reactive synthesis and syntax-guided synthesis be friends?
\newblock In Ranjit Jhala and Isil Dillig, editors, {\em 43rd {ACM} {SIGPLAN}
  Int'l Conf. on Programming Language Design and Implementation ({PLDI} 2022)},
  pages 229--243. {ACM}, 2022.

\bibitem{corsi24verification}
Davide Corsi, Guy Amir, Andoni Rodriguez, Cesar Sanchez, Guy Katz, and Roy Fox.
\newblock Verification-guided shielding for deep reinforcement learning.
\newblock {\em CoRR}, abs/2406.06507, 2024.

\bibitem{dAntoni14minimization}
Loris D'Antoni and Margus Veanes.
\newblock Minimization of symbolic automata.
\newblock In {\em Proc. of the 41st Annual {ACM} {SIGPLAN-SIGACT} Symposium on
  Principles of Programming Languages ({POPL} '14)}, pages 541--554. {ACM},
  2014.

\bibitem{dAntoni2017power}
Loris D'Antoni and Margus Veanes.
\newblock The power of symbolic automata and transducers.
\newblock In {\em Proc. of the 29th International Conference in Computer Aided
  Verification ({CAV} 2017), Part {I}}, volume 10426 of {\em LNCS}, pages
  47--67. Springer, 2017.

\bibitem{demoura08z3}
Leonardo de~Moura and Nikolaj Bj{\o}rner.
\newblock {Z3}: An efficient {SMT} solver.
\newblock In {\em Proc. of the 14th Int'l Conf. on Tools and Algorithms for the
  Construction and Analysis of Systems (TACAS'08)}, volume 4693 of {\em LNCS},
  pages 337--340. Springer, 2008.

\bibitem{faran2022ltl}
Rachel Faran and Orna Kupferman.
\newblock {LTL} with arithmetic and its applications in reasoning about
  hierarchical systems.
\newblock In {\em Proc. of the 22nd International Conference on Logic for
  Programming, Artificial Intelligence and Reasoning, ({LPAR-22.} )}, volume~57
  of {\em EPiC Series in Computing}, pages 343--362. EasyChair, 2018.

\bibitem{azadeh2017strategy}
Azadeh Farzan and Zachary Kincaid.
\newblock Strategy synthesis for linear arithmetic games.
\newblock {\em Proc. {ACM} Program. Lang.}, 2({POPL}):61:1--61:30, 2018.

\bibitem{fedyukovich19lazySynthesis}
Grigory Fedyukovich, Arie Gurfinkel, and Aarti Gupta.
\newblock Lazy but effective functional synthesis.
\newblock In {\em 20th International Conference in Verification, Model
  Checking, and Abstract Interpretation, ({VMCAI} 2019)}, volume 11388 of {\em
  LNCS}, pages 92--113. Springer, 2019.

\bibitem{finkbeiner16synthesis}
Bernd Finkbeiner.
\newblock Synthesis of reactive systems.
\newblock In Javier Esparza, Orna Grumberg, and Salomon Sickert, editors, {\em
  Dependable Software Systems Engineering}, volume~45 of {\em {NATO} Science
  for Peace and Security Series - {D:} Information and Communication Security},
  pages 72--98. {IOS} Press, 2016.

\bibitem{finkbeinerETAL2021temporalStreamLogicModuloTheories}
Bernd Finkbeiner, Philippe Heim, and Noemi Passing.
\newblock Temporal stream logic modulo theories.
\newblock In {\em 25th Int'l Conf. on Foundations of Software Science and
  Computation Structures ({FOSSACS} 2022)}, volume 13242 of {\em LNCS}, pages
  325--346. Springer, 2022.

\bibitem{gacek2015towards}
Andrew Gacek, Andreas Katis, Michael~W. Whalen, John Backes, and Darren~D.
  Cofer.
\newblock Towards realizability checking of contracts using theories.
\newblock In {\em Proc. of the 7th International Symposium {NASA} Formal
  Methods ({NFM} 2015)}, volume 9058 of {\em LNCS}, pages 173--187. Springer,
  2015.

\bibitem{geatti22linear}
Luca Geatti, Alessandro Gianola, and Nicola Gigante.
\newblock Linear temporal logic modulo theories over finite traces.
\newblock In {\em Proc. of the 31st International Joint Conference on
  Artificial Intelligence, ({IJCAI} 2022)}, pages 2641--2647. ijcai.org, 2022.

\bibitem{geatti24general}
Luca Geatti, Alessandro Gianola, and Nicola Gigante.
\newblock A general automata model for first-order temporal logics (extended
  version).
\newblock {\em CoRR}, abs/2405.20057, 2024.

\bibitem{geatti23decidable}
Luca Geatti, Alessandro Gianola, Nicola Gigante, and Sarah Winkler.
\newblock Decidable fragments of ltlf modulo theories (extended version).
\newblock {\em CoRR}, abs/2307.16840, 2023.

\bibitem{gianola2022ltl}
Alessandro Gianola and Nicola Gigante.
\newblock {LTL} modulo theories over finite traces: modeling, verification,
  open questions.
\newblock In {\em Short Paper Proceedings of the 4th Workshop on Artificial
  Intelligence and Formal Verification, Logic, Automata, and Synthesis hosted
  by the 21st International Conference of the Italian Association for
  Artificial Intelligence (AIxIA 2022)}, volume 3311 of {\em {CEUR} Workshop
  Proceedings}, pages 13--19. CEUR-WS.org, 2022.

\bibitem{heim24solving}
Philippe Heim and Rayna Dimitrova.
\newblock Solving infinite-state games via acceleration.
\newblock {\em Proc. {ACM} Program. Lang.}, 8({POPL}):1696--1726, 2024.

\bibitem{qinheping17automatic}
Qinheping Hu and Loris D'Antoni.
\newblock Automatic program inversion using symbolic transducers.
\newblock In {\em Proc. of the 38th {ACM} {SIGPLAN} Conference on Programming
  Language Design and Implementation ({PLDI} 2017)}, pages 376--389. {ACM},
  2017.

\bibitem{katis2016synthesis}
Andreas Katis, Grigory Fedyukovich, Andrew Gacek, John~D. Backes, Arie
  Gurfinkel, and Michael~W. Whalen.
\newblock Synthesis from assume-guarantee contracts using skolemized proofs of
  realizability.
\newblock {\em CoRR}, abs/1610.05867, 2016.

\bibitem{katis2018validity}
Andreas Katis, Grigory Fedyukovich, Huajun Guo, Andrew Gacek, John Backes, Arie
  Gurfinkel, and Michael~W. Whalen.
\newblock Validity-guided synthesis of reactive systems from assume-guarantee
  contracts.
\newblock In {\em Proc. of the 24th International Conference on Tools and
  Algorithms for the Construction and Analysis of Systems, ({TACAS} 2018)},
  volume 10806 of {\em LNCS}, pages 176--193. Springer, 2018.

\bibitem{maderbacher2021reactive}
Benedikt Maderbacher and Roderick Bloem.
\newblock Reactive synthesis modulo theories using abstraction refinement.
\newblock In {\em 22nd Formal Methods in Computer-Aided Design, ({FMCAD}
  2022)}, pages 315--324. {IEEE}, 2022.

\bibitem{manna95temporal}
Zohar Manna and Amir Pnueli.
\newblock {\em Temporal verification of reactive systems - safety}.
\newblock Springer, 1995.

\bibitem{meyerETAL2018strixSynthesisStrikesBack}
Philipp~J. Meyer, Salomon Sickert, and Michael Luttenberger.
\newblock Strix: Explicit reactive synthesis strikes back!
\newblock In Hana Chockler and Georg Weissenbacher, editors, {\em Computer
  Aided Verification}, pages 578--586, Cham, 2018. Springer International
  Publishing.

\bibitem{pnueli77temporal}
Amir Pnueli.
\newblock The temporal logic of programs.
\newblock In {\em Proc. of the 18th IEEE Symp. on Foundations of Computer
  Science (FOCS'77)}, pages 46--67. IEEE CS Press, 1977.

\bibitem{pnueli89onthesythesis:b}
Amir Pnueli and Roni Rosner.
\newblock On the synthesis of a reactive module.
\newblock In {\em Proc. of the 16th Annual ACM Symp. on Principles of
  Programming Languages (POPL'89)}, pages 179--190. ACM Press, 1989.

\bibitem{pnueli89onthesythesis}
Amir Pnueli and Roni Rosner.
\newblock On the synthesis of an asynchronous reactive module.
\newblock In {\em Proc. of the 16th Int'l Colloqium on Automata, Languages and
  Programming (ICALP'89)}, volume 372 of {\em LNCS}, pages 652--671. Springer,
  1989.

\bibitem{rodriguez24shield}
Andoni Rodriguez, Guy Amir, Davide Corsi, Cesar Sanchez, and Guy Katz.
\newblock Shield synthesis for {LTL} modulo theories.
\newblock {\em CoRR}, abs/2406.04184, 2024.

\bibitem{rodriguez23boolean}
Andoni Rodr\'iguez and C\'esar S\'anchez.
\newblock {Boolean Abstractions for Realizability Modulo Theories}.
\newblock In {\em Proc. of the 35th Int'l Conf. on Computer Aided Verification
  ({CAV} 2023)}, volume 13966 of {\em LNCS}, pages 1--24. Springer, 2023.

\bibitem{rodriguez24adaptive}
Andoni Rodriguez and C\'{e}sar S\'{a}nchez.
\newblock Adaptive reactive synthesis for {LTL} and {LTLf} modulo theories.
\newblock In {\em Proc. of the 38th AAAI Conf. on Artificial Intelligence
  ({AAAI} 2024)}, pages 10679--10686. {AAAI} Press, 2024.

\bibitem{rodriguez24realizability}
Andoni Rodriguez and César Sanchez.
\newblock Realizability modulo theories.
\newblock {\em Journal of Logical and Algebraic Methods in Programming}, page
  100971, 2024.

\bibitem{samuel21genSys}
Stanly Samuel, Deepak D'Souza, and Raghavan Komondoor.
\newblock Gensys: a scalable fixed-point engine for maximal controller
  synthesis over infinite state spaces.
\newblock In {\em Proc. of the 29th {ACM} Joint European Software Engineering
  Conference and Symposium on the Foundations of Software Engineering
  (ESEC/FSE'21)}, pages 1585--1589. {ACM}, 2021.

\bibitem{samuel23symbolic}
Stanly Samuel, Deepak D'Souza, and Raghavan Komondoor.
\newblock Symbolic fixpoint algorithms for logical {LTL} games.
\newblock In {\em Proc. of the 38th {IEEE/ACM} International Conference on
  Automated Software Engineering ({ASE} 2023)}, pages 698--709. {IEEE}, 2023.

\bibitem{veanes2023symbolic}
Margus Veanes, Thomas Ball, Gabriel Ebner, and Olli Saarikivi.
\newblock Symbolic automata: {\(\omega\)}-regularity modulo theories.
\newblock {\em CoRR}, abs/2310.02393, 2023.

\bibitem{walker2014predicate}
Adam Walker and Leonid Ryzhyk.
\newblock Predicate abstraction for reactive synthesis.
\newblock In {\em Proc. of the 14th Formal Methods in Computer-Aided Design,
  ({FMCAD} 2014)}, pages 219--226. {IEEE}, 2014.

\end{thebibliography}

\vfill
\pagebreak

\appendix

\section{Complete running example}

In $\phiT$ of Ex.~\ref{exRunning} a valid (positional) strategy of the system is to always play $y:2$.
In this appendix we show that a controller synthetised using our technique will,
precisely, respond in this manner infinitely many often.
To do so, we rely on the trace of Ex.~\ref{ex:provider} and Skolem functions of Ex.~\ref{ex:functions}.

First, we Booleanize $\phiT$ using \cite{rodriguez23boolean} and get $\phiB$ 
(also, recall from Ex.~\ref{exRunning} that we use the notation $c_i$ to indicate choice $i$; e.g., 
$c_0 = \{s_0,s_1,s_2\}$, $c_1 = \{s_0,s_1\}$, ...,
$c_6 = \{s_2\}$, $c_7 = \emptyset $.).
Then, we get a controller $C_{\mathbb{B}}$ from $\phiB$.
We note that many strategies satisfy $\phiB$, but $C_{\mathbb{B}}$ 
by Strix is as follows: $C_{\mathbb{B}}(e_1)=c_4$ and $C_{\mathbb{B}}(e_0)=c_1$.
Also, note that this particular strategy is memoryless, but there are
diverse strategies that use memory.
We now show how the static $\calT$-controller computes Skolem functions on demand.

\subsubsection{Step 1: Environment forces instant response.}

Let $x:4$, which holds $(x \geq 2)$ and forces constraint $(y \leq x)$. 
We are in partition $e_1$, which implies choices $\{ c_4, c_5, c_6\}$.
Now, $C_{\mathbb{B}}(e_1)=c_4$, so the $\calT$-controller looks whether the 
pair $(e_1,c_1)$ appeared before. Since it did not, it computes $\skh_{(e_1,c_4)}$ 
(see left-hand function at Ex.~\ref{ex:functions}).
Thus, $\skh_{(e_1,c_4)}(2)=2$ is the output $v_y$ in the first timestep.
Note that a $\calT$-controller with a different underlying $C_{\mathbb{B}}$ could 
also consider $c_6$ in the current play.

\subsubsection{Step 2: Environment repeats the strategy.}

Again, $x:4$ and again we are in partition $e_1$.
Now, $C_{\mathbb{B}}(e_1)=c_4$, so the $\calT$-controller looks whether the 
pair $(e_1,c_4)$ appeared before. Since it does, it just calls pre-computed $\skh_{(e_1,c_4)}$.
Thus, $\skh_{(e_1,c_4)}(2)=2$ is the output $v_y$ in the second timestep.

\subsubsection{Step 3: Environment changes its mind.}

Let $x:1$, which holds $(x < 2)$ and forces constraint $\Next(y > 1)$, whereas no constraint is further for the current timestep. 
We are in partition $e_0$, which implies choices $\{c_1, c_2\}$.
Now, $C_{\mathbb{B}}(e_0)=c_1$, so the $\calT$-controller looks whether the 
pair $(e_0,c_1)$ appeared before. Since it did not, it computes $\skh_{(e_0,c_1)}$ 
(see right-hand function at Ex.~\ref{ex:functions}).
Thus, $\skh_{(e_0,c_1)}(2)=2$ is the output $v_y$ in the third timestep.
Note that a $\calT$-controller with a different underlying $C_{\mathbb{B}}$ could 
also consider $c_2$ in the current play.

\subsubsection{Step 4: Environment prepares its trap.}

Let $x:0$, which holds $(x < 2)$ and forces constraint $\Next(y > 1)$, and take into account that the system has constraint $(y > 1)$ forced by the previous timestep. 
We are again in partition $e_0$, which implies, again, choices $\{c_1, c_2\}$.
Now, $C_{\mathbb{B}}(e_0)=c_1$, so the $\calT$-controller looks whether the 
pair $(e_0,c_1)$ appeared before. Since it does, it just calls pre-computed $\skh_{(e_0,c_1)}$.
Thus, $\skh_{(e_0,c_1)}(2)=2$ is the output $v_y$ in the fourth timestep.
Note that this time there is no correct $C_{\mathbb{B}}$ that could 
also consider $c_2$ in the current play.

\subsubsection{Step 5: Environment strikes back!}

Let $x:2$, which holds $(x \geq 2)$ and forces constraint $(y \leq x)$. 
Also, note that the system has constraint $(y>1)$ from previous timestep. 
We are in partition $e_1$, which implies choices $\{ c_4, c_5, c_6\}$.
Now, the same as in step 2 happens: $C_{\mathbb{B}}(e_1)=c_4$, so the $\calT$-controller looks whether the 
pair $(e_1,c_4)$ appeared before. Since it does, it just calls pre-computed $\skh_{(e_1,c_4)}$.
Thus, $\skh_{(e_1,c_4)}(2)=2$ is the output $v_y$ in the fifth timestep.

Note that this is the dangerous situation, where the system can only output $y:2$;
in other words, it happens again that there is no correct $C_{\mathbb{B}}$ that could 
also consider another choice (in this case, $c_4$ and $c_5$) in the current play.
We can derive all the possible behaviours from these steps.
Also, note that another possibility is to pre-compute all the Skolem functions,
but it is less efficient.

\section{More About Adaptivity}

We outline several immediate consequences of using adaptivity of Sec.~\ref{sec:adapt}. 

%\subsection{Reactive Synthesis Beyond $\LTLt$}

\subsubsection{Across-time Adaptivity.}

In the previous section we showed that if we can synthesize Skolem
functions for adaptive provider formulae, then the system obtained is
still a good system for the $\phiT$.
We show now that the additional arguments can be used to produce
better controllers for $\phiT$.
For instance, $\zs$ (thus, $\vzs$ in Fig.~\ref{figAdaptive}) can be used to feed past values to the controller,
and $\psi^+$ can describe desired evolution
of the output in terms of the past history.

\begin{definition} [Across-time adaptive controller]%
  \label{defAcrossAdapt}%
  Let $\phiT(\xs,\ys)$ be a specification, let $\extraCons$ be an
  adaptive provider description, and let $\rhoT_\Gamma$ be the resulting
  controller.
  Then, we say that $\rhoT_\Gamma$ is an across-time adaptive controller if
  the extra variables $\zs$ in $\extraCons$ fed are past values of
  $\xs$ and $\ys$.
\end{definition}

% 
% Example
%
\begin{example}
  Consider again the characteristic formula $\psi = (y>x)$, and the
  corresponding basic provider formula $\forall x. \exists y. \psi$,
  which is valid in $\ThZ$.
  The constraint $\psi^+ = (y>z)$ makes the adaptive provider formula
  $\varphi = \forall x,z. \exists y. (\psi \wedge \psi^+)$ valid.
  A Skolem function $\skh(x,z)$ guarantess that the output $y$ generated
  is greater than the values of both $x$ and $z$.
  Then, if the controller received $z$ as the value of $y$ in the
  previous timestep (denoted $y^Y$), then we have that the controller
  will generate outputs that are monotonically increasing.

  Note that time adaptivity is not always possible.
  For example, using the constraint $\psi^+_2=(y<z)$ (which would
  force the output to be monotonically decreasing) would turn the
  resulting adaptive provider formula invalid, so a controller cannot
  be produced.
\end{example}  

%
% Note: LTLt not expressive enough
%
The practical implication of Def.~\ref{defAcrossAdapt} is that we can
now produce controllers that in practice satisfy constraints that were
not expressible in $\LTLt$ before, because the transfer of values
across time quickly leads to undecidability of the realizability
problem.
Note that the Booleanization process in Sec.~\ref{sec:static} only
considers $\LTLt$ specifications where literals do not relate
variables from different time instants (also called \emph{non-cross
  state} fragment in \cite{geatti23decidable}).
%
% Note: statistics of the history
%
Note that across-time adaptivity can be also used feeding to $\zs$ the
result of evaluating functions like \texttt{average},
\texttt{historyMax} on past values of $\xs$ and $\ys$.

%
% Repetitive: not expressible in LTLt
%
% The intention is that the designer provides $\phiT$ and
% $\extraCons=\{\psi^+\}$ and obtains a controller $\rhoT_\Gamma$ whose
% behavior is not synthesizable using the method of
% Sec.~\ref{sec:static} from any suitable $\phiT$.

% \begin{lemma}
%   Given $\phiT$ and $\extraCons$, it is decidable whether there exists
%   and across-time adaptive controller for $\phiT$ and $\extraCons$.
% \end{lemma}

% \begin{proof}
%   (Sketch). By Cor.~\ref{cor:decidability}, given $\phiT$ and for all
%   elements $\psi^+$ in $\extraCons$, $\psi^+ \in \calT$, then the
%   adaptive reactive synthesis problem is decidable for $\phiT$ (under
%   the $\calT$ $\exists^*\forall^*$ decidability of $\calT$) and
%   produces $\rhoT(\extraCons)$.
% %
% The same holds for $\rhoT(\extraCons)$ that are across-time adaptive.
% \end{proof}

As another example, a \provider that is only required to keep $y$
within limits $[0, 0.3]$ may choose any value within the bounds with a
fixed Skolem function.
However, the designer may prefer to choose ``smooth'' values that do
not change dramatically to avoid abrupt changes in the values of $y$
(see Fig.~\ref{figSmoothSignal}).
Therefore,
\begin{wrapfigure}[10]{l}{0.40 \textwidth}
  \begin{minipage}{0.40 \textwidth}
    \vspace{-4em}
\begin{figure}[H]
  \includegraphics[width=1\linewidth]{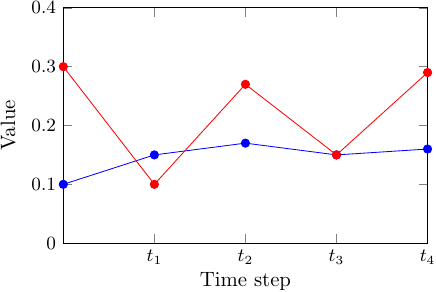}
    \vspace{-2.0 em}
  \caption{Abrupt vs. desired}
  \label{figSmoothSignal}
\end{figure}
\end{minipage}
%\vspace{-2em}
\end{wrapfigure}
instead of synthetising a basic function $f(x)$ for the
valid formula $\forall x. \exists y. (0 \leq y \leq 3)$, the
controller will produce another $f(x)^{\text{smooth}}$ function for
the valid formula
$\forall x,y^Y. \exists y. (0 \leq y \leq 3) \wedge \neg \exists
w. [((0 \leq z \leq 3)) \wedge (|w-y^Y|<|y-y^Y|)]$, so that produced
value for $y$ is as close as possible to value of $y$ in the previous
timestep.
Since the function exists, then the enriched specification will
produce the desired controller that spans out of $\LTLt$.
Note that the constraint $\psi^+$ in the previous example uses an
additional quantifier $\exists w$.
We show next practical applications and how to cope with this inner
quantifier.

%
%In the platooning use case above all the consequents of the
%requirements (except of $\varphi_0$) could provide with a safe range
%of acceleration values, but did not rather provide with the function to
%calculate an exact value (which we can now, since these function can use previous values).

\subsubsection{Approximations and Neurosymbolic Control.}
%

%\noindent
%\fbox{
%\begin{minipage}{1\textwidth}
%\emph{
Skolem functions computed from basic provider formulae have a
shape $\forall^*\exists^*. \psi$.
This shape is preserved in adaptive provider formulae in which the
constraint $\psi^+$ is quantifier-free.
However, as the last example illustrated, the constraint $\psi^+$ may
include quantifiers, which does not preserve the shape typically
amenable for Skolemization.

For instance, to compute the smallest $y \in \dom(\ThZ)$ in
$\forall x. \exists y. (y>x)$, one can use the adaptive provider
formula
$\forall x. \exists y. [(y>x) \wedge \forall z. (z>x) \Into
(z\geq{}y)] $.
We overcome this issue by performing quantifier elimination (QE) for
the innermost quantifier and recover the $\forall^*\exists^*$
shape.
In consequence, our resulting method works on any theory $\calT$ that:
\begin{compactenum}[(1)]
\item is decidable for the $\exists^*\forall^*$ fragment (for the
  Boolean abstraction);
\item permits a Skolem function synthesis procedure (for valid
  $\forall^*\exists^*$ formulae), for producing static providers; and
\item accepts QE (which preserves formula equivalence) for the
  flexibility in defining quentified constraints $\psi^+$.
\end{compactenum}
%\end{minipage}
%}

The use of quantification opens the door to explore solutions that
exploit characteristics of concrete theories.
Consider for example the theory of Presburger arithmetic $\ThZ$ (which
we illustrate with single variable but can be extended to other
notions of distance with multiple variables, such as Euclidean
distance).
In this theory the following holds.

\begin{lemma}[Closest element]
  \label{lem:Zclosest}
  In $\ThZ$ the following holds.
  Assuming $\forall\xs.\exists y.\psi(\xs,y)$, the following is also
  valid:
  \[
    \forall\xs.z.\exists y.\big(\psi(\xs,y) \And \forall w.[\psi(\xs,w)\Into |y-z|\leq |w-z|]).
  \]
\end{lemma}
In other words, in $\ThZ$ if for all inputs $\xs$ there is an output
$y$ such that $\psi$ holds, then there is always a closest value to
any provided $z$ that satisfies $\psi$.
The method we propose uses QE to provide a quantifier-free
formula equivalent to
$\psi^+=\forall w. [\psi(\xs,w)\Into |y-z|\leq |w-z|])$.
The formula generated by the elimination depends on each specific
constraint.
Note that this formula only has $\xs$, $z$ and $y$ as free variables.
As a result the Skolem function generation will produce a function
$\skh(\xs,z)$ such that, given values for $\xs$ and $z$, will provide an
appropriate output (for $y$) that satisfies $\psi$ and is minimal (as
expressed by the constraint $\psi^+$).

The following holds in $\ThZ$, which essentially states that if the
candidate $z$ satisfies $\psi$ then $\skh$ will return it, and that $\skh$
always returns a value satisfies $\psi$ and that is as close to $z$ as
possible.

\begin{proposition}
  Let $\psi(\vxs,y)$ be valid and let $\skh$ be a Skolem function of
  \[
    \forall\xs.z.\exists y.\big(\psi(\xs,y) \And \forall w.[\psi(\xs,w)\Into |y-z|\leq |w-z|]).
  \]
  Let $\vxs\in\val(\xs)$ and $v_z\in\val(z)$ be arbitrary value,
  and let $v_y=\skh(\vxs,v_z)$.
  Then,
  \begin{compactitem}
  \item If $\psi(\vxs,v_z)$ then $v_y=v_z$.
  \item Let $o$ be such that $\psi(\vxs,o)$. Then $|v_z-v_y|\leq |v_z-o|$. 
  \end{compactitem}
\end{proposition}

The theory of real arithmetic $\ThR$ is also widely used, but
unfortunately the equivalent of Lemma~\ref{lem:Zclosest} does not hold
because it may not be possible to find a value that satisfies $\psi$
and is the closest to a given candidate.
However, it is always possible to compute a closest within a given
tolerance given by a real constant $\epsilon$.
This is expressed in the following lemma.

\begin{lemma}[Approximately Closest Element]
  \label{lem:Rclosest}
  In $\ThR$ the following holds.
  Assuming $\forall\xs.\exists y.\psi(\xs,y)$, then for every constant
  $\epsilon>0$, the following is valid:
  \[
    \forall\xs.z.\exists y.\big(\psi(\xs,y) \And \forall w.[\psi(\xs,w)\Into |y-z|\leq |w-z|+\epsilon]).
  \]
\end{lemma}

%
% Note, alternative to eps
%
Note that there is a practical alternative to using a constant
$\epsilon$ (provided by the user), which consists of iteratively
generating Skolem functions for increasingly better results, by
starting from $\forall x. \exists z.\exists y. \psi(x,y)$ and refining
to better Skolem functions with respect to e.g., closest $y$,
such that, at every iteration, the resulting provider guarantees $\phiT$.
%
%This is described in App.~\ref{app:iterative}.

% %
% % External program chooses constraints
% %
% Both across-time adaptivity and metric-spaces allow to synthetise
% controllers are running in environments where resources have to be
% dedicatedly managed and thus it can use some oracle (for example a
% program) that chooses $\extraCons$ dynamically.
% %
% % 
% %
Lemmas~\ref{lem:Zclosest} and~\ref{lem:Rclosest} allow computing
adaptive controllers that guarantee $\phiT$ and approximate the value
provided externally to the controller.
A very interesting possibility is to use an external program,
e.g., computed using machine-learning (ML), that provides 
$\vzs$ of Fig.~\ref{figAdaptive}.
We call {neurosymbolic reactive synthesis} to this combination of ML
 and adaptive controller synthesis.
In this manner, we combine a correct-by-construction technique (the
reactive synthesis modulo theories presented here) with a richer
but potentially incorrect program that is trained to optimize sophisticated goals.
The result is a controller that produces safe outputs by
approximating the values proposed by the unsafe ML controller.
This approach resembles shielding
\cite{alshiekhETAL2017safeReinforcementLearningShielding} in the sense
that our approach will also always produce safe outputs.
However, in our approach the value chosen can be different to the
proposed by the ML even if the value proposed is safe (for example, to
also guarantee smoothness).
%
%A complete exploration of the principle of neurosymbolic reactive
%synthesis is out of the scope of this paper.
%
A further exploration of applications of neurosymbolic synthesis is out of the scope of this paper.
%
% Although this paper does not seek to further explore applications of this,
% it positions theoretical machinery for both neurosymbolic synthesis and 
% shielding of $\LTLt$ (and enriched) specifications. 
\iffalse
\begin{figure}[b!]
\centering
\begin{minipage}{.5\textwidth}
  \centering
  \includegraphics[width=1\linewidth]{figures/dualShield.pdf}
  \captionof{figure}{Neurosymbolic synthesis}
  \label{figNeurosafe}
\end{minipage}%
\begin{minipage}{.5\textwidth}
  \centering
  \includegraphics[width=.8\linewidth]{figures/shield.pdf}
  \captionof{figure}{Shielding~\cite{alshiekhETAL2017safeReinforcementLearningShielding}.}
  \label{figShield}
\end{minipage}
\end{figure}
\fi

%%% Local Variables:
%%% TeX-master: "main.tex"
%%% TeX-PDF-mode: t
%%% End:

%\newpage

%\input{appendix}

\end{document}